\documentclass[reprint,amsmath,amssymb,aps,prl,superscriptaddress, longbibliography]{revtex4-2}
\usepackage{hyperref,url} \usepackage{amsmath,bm} \usepackage{mathtools}
\usepackage{amsfonts}
\usepackage{cleveref}
\hypersetup{breaklinks,colorlinks=true,linkcolor=blue,citecolor=blue,urlcolor=blue,filecolor=blue}
\usepackage[dvipsnames]{xcolor} \usepackage{graphicx} \usepackage{booktabs}
\usepackage{physics}

\usepackage[english]{babel}
\makeatletter
\let\ORIbbl@fixname\bbl@fixname
\def\bbl@fixname#1{%
  \@ifundefined{languagealias@\expandafter\string#1}
    {\ORIbbl@fixname#1}
    {\edef\languagename{\@nameuse{languagealias@#1}}}%
}
\newcommand{\definelanguagealias}[2]{%
  \@namedef{languagealias@#1}{#2}%
}
\makeatother

\definelanguagealias{en}{english}
\usepackage{soul}

\usepackage{amsthm} 

\usepackage{algorithmicx,algorithm}
\usepackage[noend]{algpseudocode}
\usepackage{tabularx}

\definecolor{darkmagenta}{RGB}{150,36,153}
\definecolor{darkindigo}{RGB}{80,44,209}
\definecolor{darkred}{RGB}{180,24,110}
\definecolor{themeblue}{RGB}{26,76,200}
\definecolor{themepink}{RGB}{227,45,145}
\definecolor{thememagenta}{RGB}{200,48,204}
\definecolor{lightindigo}{RGB}{137,113,225}
\definecolor{lightblue}{RGB}{0,176,240}

\newcommand{\eq}[1]{Eq.~\hyperref[eq:#1]{(\ref*{eq:#1})}}
\newcommand{\citen}[1]{Ref.~\citenum{#1}}
\renewcommand{\sec}[1]{\hyperref[sec:#1]{Section~\ref*{sec:#1}}}
\newcommand{\app}[2]{\hyperref[app:#1]{Appendix~\ref*{app:#1}}}
\newcommand{\tab}[1]{\hyperref[tab:#1]{Table~\ref*{tab:#1}}}
\newcommand{\fig}[1]{\hyperref[fig:#1]{Figure~\ref*{fig:#1}}}
\newcommand{\thm}[1]{\hyperref[thm:#1]{Theorem~\ref*{thm:#1}}}
\newcommand{\cor}[1]{\hyperref[cor:#1]{Corollary~\ref*{cor:#1}}}
\newcommand{\defi}[1]{\hyperref[def:#1]{Definition~\ref*{def:#1}}}

\newtheorem{theorem}{Theorem} 
\newtheorem{lemma}[theorem]{Lemma}
\newtheorem{corollary}[theorem]{Corollary}
\theoremstyle{definition}
\newtheorem{definition}{Definition}

\newcommand{\eps}{\varepsilon}
\newcommand{\E}{\mathbb{E}}

\newcommand{\idmat}[0]{\mathbb{I}}

\newcommand{\bigot}[1]{\widetilde{\mathcal{O}}( #1)}

\newcommand{\bigo}[1]{\mathcal{O}(#1)}

\newcommand{\bigomega}[1]{\Omega(#1)}
\newcommand{\defeq}[0]{\coloneqq}
\newcommand{\controlled}[1]{\textsc{c-}#1}

\newcommand{\Google}{\affiliation{Google Quantum AI, Mountain View, CA, USA}}

\begin{document}
\nocite{apsrev41control}

\author{William J. Huggins}
\email{corresponding author: whuggins@google.com}
\Google

\author{Kianna Wan}
\Google
\affiliation{Stanford Institute for Theoretical Physics, Stanford University, Stanford, CA 94305, USA}

\author{Jarrod McClean}
\Google

\author{Thomas E.~O'Brien}
\Google

\author{Nathan Wiebe}
\email{corresponding author: nawiebe@cs.toronto.edu}
\affiliation{University of Toronto, Toronto, ON, CA}
\affiliation{Pacific Northwest National Laboratory, Richland, WA, USA}

\author{Ryan Babbush}
\email{corresponding author: babbush@google.com}
\Google

\title{Nearly Optimal Quantum Algorithm for Estimating Multiple Expectation Values}

\begin{abstract}
  Many quantum algorithms involve the evaluation of expectation values.
  Optimal strategies for estimating a single expectation value are known,
  requiring a number of state preparations that scales with the target error $\varepsilon$
  as $\mathcal{O}(1/\varepsilon)$.
  In this paper, we address the task of estimating the expectation values of
  \(M\) different observables, each to within additive error \(\varepsilon\), with the
  same \(1/\varepsilon\) dependence.
  We describe an approach that leverages Gily\'{e}n \emph{et al.}'s~quantum
  gradient estimation algorithm to achieve $\mathcal{O}(\sqrt{M}/\varepsilon)$
  scaling up to logarithmic factors, regardless of the commutation properties of
  the $M$ observables.
  We prove that this scaling is worst-case optimal in the high-precision regime if the state preparation is treated as a black box, even when the operators are mutually commuting.
  We highlight the flexibility of our approach by presenting several
  generalizations, including a strategy for accelerating the estimation of a
  collection of dynamic correlation functions.
\end{abstract}
\maketitle

\subsection*{Introduction}

A fundamental task of quantum simulation is to perform an
experiment \textit{in silico}.
Like traditional experimentalists, researchers using
quantum computers will often be interested in efficiently measuring a collection of properties.
For example, the electronic ground state problem is frequently cited as a
motivation for quantum simulation of chemistry, but determining the ground state
energy is only a starting point in most chemical applications.
Depending on context, it may be essential to measure the dipole moment and
polarizability, the electron density, the forces experienced by the classical
nuclei, or various other quantities~\cite{Pulay1979-gd, Gregory1997-yy}.
Similarly, in condensed matter physics and beyond, correlation functions play a
central role in the theory of quantum many-body phenomena due to
their interpretability and measurability in the lab~\cite{Damascelli2004-yk, Rickayzen2013-xj}.

In this letter, we consider the problem of accurately and efficiently estimating
multiple properties from a quantum computation. We focus on evaluating the
expectation values of a collection of \(M\) Hermitian operators \(\{O_j\}\) with
respect to a pure state \(\ket{\psi}\). We aim to evaluate each expectation
value to within additive error \(\varepsilon\) using as few calls as possible
to a state preparation oracle for \(\ket{\psi}\) (or its inverse). One simple
approach is to repeatedly prepare \(\ket{\psi}\) and projectively measure mutually commuting subsets of \(\{O_j\}\). Alternatively, strategies
based on amplitude estimation achieve a quadratic speedup with respect to
\(\varepsilon\) but entail measuring each observable separately~\cite{Brassard2000-oi, Knill2007-ci, Rall2020-in}. A range of newer
``shadow tomography'' techniques use joint measurements of multiple copies
of $\ket{\psi}$ to achieve polylogarithmic scaling with respect to \(M\) at the
expense of an unfavorable \(1/\varepsilon^4\) scaling~\cite{Aaronson2020-ei,
  Brandao2017-gl, Van_Apeldoorn2019-xl, Huang2021-pr}. In certain situations,
randomized methods based on the idea of ``classical shadows'' of the state obtain
\(1/\varepsilon^2\) scaling while improving upon sampling
protocols with deterministic measurement settings~\cite{Huang2020-kn, Zhao2021-fv}. We review these existing approaches in \app{prior_estimation_work}{I} and compare them to our new strategy in \tab{cost_comparison} and \app{applications}{II}.

Our main contribution is an algorithm that achieves the
same \(1/\varepsilon\) scaling as methods based on amplitude estimation, but
also improves the scaling with respect to \(M\) from \(\bigot{M}\)
to \(\bigot{\sqrt{M}}\), where the tilde in \(\bigot{\cdot}\) hides logarithmic factors.
Our approach is to construct a function $f$ whose gradient yields the
expectation values of interest and encode $f$ in a
parameterized quantum circuit.
We can then apply Gily\'{e}n \emph{et al.}'s quantum algorithm for gradient
estimation~\cite{Gilyen2017-gk} to obtain the desired scaling.
The following theorem formalizes our result.
\begin{theorem}\label{thm:ub}
  Let $\{O_j\}$ be a set of $M$ Hermitian operators on \(N\) qubits, with
  spectral norms \(\|O_j\| \leq 1\) for all \(j\).
  There exists a quantum algorithm that, for any $N$-qubit quantum state
  $\ket{\psi}$ prepared by a unitary $U_\psi$, outputs estimates
  $\widetilde{o_j}$ such that $ | \widetilde{o_j}
    - \bra{\psi} O_j \ket{\psi}| \le \varepsilon$ for all \(j\) with probability at
  least $2/3$, using \(\bigot{\sqrt{M}/\varepsilon}\) queries to $U_\psi$ and
  $U_\psi^\dagger$, along with $\bigot{\sqrt{M} /\varepsilon}$ gates of the form
  controlled-$e^{-ix O_j}$ for each \(j\), for various values of $x$ with \(|x| \in
  \bigo{1/\sqrt{M}}\).
\end{theorem}

As we show in \cor{lower_bounds}, this query complexity is worst-case
optimal (up to logarithmic factors) in the high-precision regime where \(\varepsilon \in (0, \frac{1}{3\sqrt{M}})\).
After establishing this lower bound for our problem, we review the gradient
algorithm of \citen{Gilyen2017-gk} and present the proof of \thm{ub}.
We then discuss several extensions of our approach, including a strategy for estimating multiple dynamic correlation functions and a method that handles
observables with arbitrary norms (or precision requirements) based on a
generalization of the gradient algorithm.

\begin{table}[]
  \begin{tabularx}{.483\textwidth}{@{}lXXX@{}}
    \textbf{\;\;\;\;\;\;\;\;\;\;\;\;\;\;} & \textbf{Comm.}                           & \textbf{Non-comm.}                       & \textbf{\(k\)-RDM}                                       \\ \toprule
    Sampling                              & \(\bigo{\frac{\log M}{\varepsilon^2}}\)  & \(\bigot{\frac{M}{\varepsilon^2}}\)      & \(\bigot{\frac{N^k}{\varepsilon^2}}\)~\cite{Zhao2021-fv} \\ \midrule
    Amp. Est.~\cite{Knill2007-ci}         & \(\bigot{\frac{M}{\varepsilon}}\)        & \(\bigot{\frac{M}{\varepsilon}}\)        & \(\bigot{\frac{N^{2k}}{\varepsilon}}\)                 \\ \midrule
    Shadow Tom.~\cite{Huang2021-pr}       & \(\bigo{\frac{\log M}{\varepsilon^4}}\)  & \(\bigo{\frac{\log{M}}{\varepsilon^4}}\) & \(\bigo{\frac{k \log{N}}{\varepsilon^4}}\)               \\ \midrule \midrule
    Gradient                              & \(\bigot{\frac{\sqrt{M}}{\varepsilon}}\) & \(\bigot{\frac{\sqrt{M}}{\varepsilon}}\) & \(\bigot{\frac{N^k}{\varepsilon}}\)                      \\
  \end{tabularx}
  \caption{ A comparison of the (worst-case) complexities, in terms of state
    preparation oracle queries, of different approaches for measuring multiple
    observables. We consider three applications: estimating the expectation
    values of $M$ commuting or non-commuting observables, and
    determining the fermionic \(k\)-RDM of an \(N\)-mode system. Here,
    \(\varepsilon\) denotes the additive error to which each quantity is
    estimated. We compare strategies based on naive sampling, amplitude
    estimation, and shadow tomography to our gradient-based approach. 
    We cite the
    specific works used to determine these complexities, including the
    Pauli-specific shadow protocol of~\citen{Huang2021-pr}. Note that methods based on sampling and shadow tomography also work under a weaker input model where only copies of the state are provided.
  }
  \label{tab:cost_comparison}
\end{table}

\subsection*{Lower Bounds}
\label{sec:lower_bounds}

In \citen{Van_Apeldoorn2021-sk}, Apeldoorn proved a lower bound for a task that
is essentially a special case of our quantum expectation value problem.
We explain how a lower bound for our problem can be obtained as a
corollary.
Their results are expressed in terms of a particular quantum access model for
classical probability distributions:
\begin{definition}[Sample oracle for a probability distribution]
  \label{def:probability_oracle_classical_distribution}
  Let \(\vb{p}\) be a probability distribution over \(M\) outcomes, i.e.,
  \(\vb{p} \in [0,1]^M\) with \(\|\vb{p}\|_1 = 1\).
  A \textit{sample oracle} \(U_{\vb{p}}\) for \(\vb{p}\) is a unitary operator that acts
  as
  \begin{equation}
    \label{eq:multi_dimensional_probability_oracle}
    U_{\vb{p}}: \ket{0}\ket{0} \mapsto \sum_{j=1}^M \sqrt{p_j} \ket{j} \otimes \ket{\phi_j},
  \end{equation}
  where the \(\ket{\phi_j}\) are arbitrary normalized quantum states.
\end{definition}

We rephrase Lemma 13 of \citen{Van_Apeldoorn2021-sk} below.
Here and throughout this paper, we count queries to a unitary oracle \(U\) and to
its inverse \(U^\dagger\) as equivalent in cost.
\begin{theorem}[Lemma 13,~\citen{Van_Apeldoorn2021-sk} (rephrased)]
  \label{thm:lower_bound_apeldoorn}
  Let \(M\) be a positive integer power of \(2\) and let \(\varepsilon \in (0, \frac{1}{3\sqrt{M}})\).
  There exists a known matrix \(A \in \{-1,+1\}^{M \times M}\) such that the
  following is true.
  Suppose \(\mathcal{A}\) is an algorithm that, for every probability distribution
  \(\vb{p}\), accessed via a sample oracle \(U_{\vb{p}}\), outputs (with
  probability at least \(2/3\)) a \(\vb{\tilde{q}}\) such that \(\|A \vb{p} -
  \vb{\tilde{q}}\|_\infty \leq \varepsilon\).
  Then \(\mathcal{A}\) must use \(\bigomega{{\sqrt{M}}/{\varepsilon}}\) queries to \(U_{\vb{p}}\) in the worst case.
\end{theorem}

We can use this theorem to derive the following corollary, establishing the near-optimality of the algorithm in~\thm{ub} in certain regimes.
\begin{corollary}
  \label{cor:lower_bounds}
  Let \(M\) be a positive integer power of \(2\) and let \(\varepsilon \in (0, \frac{1}{3\sqrt{M}})\).
  Let \(\mathcal{A}\) be any algorithm that takes as an input an arbitrary
  set of \(M\) observables \(\{O_j\}\).
  Suppose that, for every quantum state \(\ket{\psi}\), accessed via a state
  preparation oracle \(U_{\psi}\), \(\mathcal{A}\) outputs estimates of each
  \(\ev{O_j}{\psi}\) to within additive error \(\varepsilon\) (with probability at
  least \(2/3\)).
  Then, there exists a set of observables \(\{O_j\}\) such that \(\mathcal{A}\)
  applied to \(\{O_j\}\) must use \(\bigomega{{\sqrt{M}}/{\varepsilon}}\)
  queries to \(U_{\psi}\).
\end{corollary}

\begin{proof}
  Assume for the sake of contradiction that for any $\{O_j\}$ and $U_\psi$, the algorithm $\mathcal{A}$ uses $o(\sqrt{M}/\varepsilon)$
  queries to $U_\psi$ to estimate every $\bra{\psi}O_j\ket{\psi}$ to within
  error $\varepsilon$ (with success probability at least $2/3$).
  For any sample oracle $U_{\vb{p}}$ of the form in
  \eq{multi_dimensional_probability_oracle}, consider the state
  \begin{equation}
    \ket{\psi(U_{\vb{p}})} \defeq \sum_{j=1}^M \sqrt{p_j} \Big( \bigotimes_{i=1}^M \ket{\frac{1 - A_{ij}}{2}}\Big) \otimes \ket{j} \otimes \ket{\phi_j}.
  \end{equation}
  A quick computation verifies that the $i$-th entry of the vector $A\vb{p}$ is
  equal to \(\ev{Z_i}{\psi(U_{\vb{p}})}\), where $Z_i$
  denotes the Pauli $Z$ operator acting on the $i$-th qubit.
  Since the matrix $A$ is known, it is clear that $\ket{\psi(U_{\vb{p}})} = U_A
    (I\otimes U_{\vb{p}}) \ket{0}$ for a known unitary $U_A$:
  \begin{equation}
    \label{eq:O_ap}
    U_{A} = \sum_{j} \big(\bigotimes_{i=1}^M X_{i}^{\delta_{A_{ij}, -1}}\big) \otimes \ketbra{j} \otimes \idmat.
  \end{equation}
  Therefore, we can apply algorithm \(\mathcal{A}\) with \(O_j = Z_j\) for \(j
  \in \{1, \cdots, M\}\) and \(U_\psi = U_A (\idmat \otimes U_{\vb{p}})\).
  By our assumption, this constitutes an algorithm that for every $U_{\vb{p}}$, estimates each entry of \(A\vb{p}\) to within
  error \(\varepsilon\) using \(o(\sqrt{M}/\varepsilon)\)
  queries to \(U_{\vb{p}}\), contradicting \thm{lower_bound_apeldoorn}, and
  completing the proof.
\end{proof}

\subsection*{Background on Gily\'{e}n et al.'s gradient algorithm}

Our framework for simultaneously estimating multiple expectation values uses the improved quantum algorithm for gradient estimation of
Gily\'{e}n, Arunachalam, and Wiebe (henceforth, Gily\'{e}n \emph{et
  al.})~\cite{Gilyen2017-gk}.
Gily\'{e}n \textit{et al.}~built on earlier work by Jordan~\cite{Jordan2005-hs},
which demonstrated an exponential quantum speedup for computing the gradient in
a particular black-box access model.
Specifically, Jordan's algorithm uses one query to a \textit{binary oracle} (see
\app{more_gradient_details}{III}) for a function $f$, along with phase
kickback and the quantum Fourier transform, to obtain an approximation
of the gradient $\nabla f$.

While we defer a technical discussion of Gily\'{e}n \emph{et al.}'s algorithm to \app{more_gradient_details}{III} (and refer the reader also to~\citen{Gilyen2017-gk}), we give a brief, colloquial description of their algorithm here.
It is helpful to review their definition for a \textit{probability oracle},
\begin{definition}[Probability oracle]
  \label{def:probability_oracle_main_text}
  Consider a function \(f: \mathbb{R}^M \rightarrow [0, 1]\).
  A probability oracle \(U_f\) for \(f\) is a unitary operator that acts
  as
  \begin{align} \label{eq:U_f}
    U_f:& \ket{\bm{x}}\ket{\bm{0}} \mapsto \\
    &\ket{\bm{x}}\left(\sqrt{f(\bm{x})}  \ket{1} \ket{\phi_1(\bm{x})} + \sqrt{1 - f(\bm{x})}\ket{0}\ket{\phi_0(\bm{x})}\right), \nonumber
  \end{align}
  where \(\ket{\bm{x}}\) denotes a discretization of the variable \(\bm{x}\)
  encoded into a register of qubits, \(\ket{\bm{0}}\) denotes the all-zeros state of a register of ancilla qubits, and \(\ket{\phi_0(\bm{x})}\) and \(\ket{\phi_1(\bm{x})}\) are
  arbitrary quantum states.
\end{definition}
Gily\'{e}n \textit{et al.} show how such a probability oracle can be used to encode a finite-difference approximation to a directional derivative of \(f\) in the phase of an ancilla register, e.g., a first-order approximation is implemented by
\begin{equation}
    A_{f'_1}: \ket{\bm{x}}\ket{\bm{0}} \mapsto e^{i \left(f\left(\bm{x}\right) - f\left(\bm{-x}\right)\right)} \ket{\bm{x}}\ket{\bm{0}}.
\end{equation}
As in Jordan's original algorithm, a quantum Fourier transform can then be used to extract an approximate gradient from the phases accumulated on an appropriate superposition of basis states. By using higher-order finite-difference formulas,
Gily\'{e}n \textit{et al.} are able to estimate the gradient with a scaling that is optimal (up to logarithmic factors) for a particular family of smooth functions.
We restate the formal properties of their algorithm in the theorem below. 

\begin{theorem}[Theorem 25, \citen{Gilyen2017-gk} (rephrased)]
  \label{thm:gradient_algorithm}
  Let \(\varepsilon\), \(c \in \mathbb{R}_{+}\) be fixed constants, with \(\varepsilon
  \leq c\).
  Let $M \in \mathbb{Z}_{+}$ and $\bm{x} \in \mathbb{R}^M$.
  Suppose that \(f: \mathbb{R}^M \rightarrow \mathbb{R}\) is an analytic
  function such that for every \(k \in \mathbb{Z}_+\), the following bound holds
  for all \(k\)-th order partial derivatives of \(f\) at $\bm{x}$ (denoted by
  \(\partial_{{\bm{\alpha}}} f(\bm{x})\)): $|\partial_{{\bm{\alpha}}} f(\bm{x})
    | \leq c^kk^{\frac{k}{2}}$.
  Then, there is a quantum algorithm that outputs an estimate
  \(\widetilde{\bm{g}} \in \mathbb{R}^M\) such that $\|\nabla f(\bm{x}) -
    \widetilde{\bm{g}}\|_\infty \leq \varepsilon$, with probability at least \(1 -
  \delta\).
  This algorithm makes \(\bigot{c \sqrt{M} \log(M/\delta)/\varepsilon }\) queries
  to a probability oracle for \(f\).
\end{theorem}
\subsection*{Expectation values via the gradient algorithm}
\label{sec:speedup_by_gradient}

To construct our algorithm and prove \thm{ub}, we build a probability oracle for a function whose gradient encodes the expectation values of interest and apply the quantum algorithm for the gradient.
\begin{proof}[Proof of~\thm{ub}]
  We begin by defining the parameterized unitary
  \begin{equation}
    \label{eq:unitary_of_x}
    U(\bm{x}) \defeq \prod_{j=1}^M e^{-2 i x_j O_j}
  \end{equation}
  for $\bm{x} \in \mathbb{R}^M$.
  The derivative of this unitary with respect to \(x_\ell\) is
  \begin{equation}
    \label{eq:d_unitary_of_x}
    \frac{\partial U}{\partial x_\ell} =  -2 i \Bigg(\prod_{j=1}^\ell e^{-2 i x_j O_j}\Bigg) O_\ell \Bigg( \prod_{k=\ell + 1}^M e^{-2 i x_k O_k}\Bigg).
  \end{equation}
  We are interested in the expectation of the \(O_j\) with respect to the state
  \(\ket{\psi}\), so we define the following function \(f\):
  \begin{equation}
    \label{eq:f_def}
    f(\bm{x}) \defeq -\frac{1}{2}\mathrm{Im}[\bra{\psi} U(\bm{x})\ket{\psi}] +\frac{1}{2}.
  \end{equation}
  Using \eq{d_unitary_of_x}, we have 
  \begin{equation}
    \label{eq:derivative_magic_evaluation}
    \frac{\partial f}{\partial x_\ell}\Bigg|_{\bm{x} = \bm{0}} = \ev{O_\ell}{\psi}.
  \end{equation}
  Therefore, the gradient $\nabla f(\bm{0})$ is precisely the collection of expectation values of interest.

  Now, we verify that \(f\) satisfies the conditions of
  \thm{gradient_algorithm}.
  Observe that $f$ is analytic and that the \(k\)-th order partial derivative of \(f\) with respect to any
  collection of indices \(\alpha \in \{1,\dots, M\}^k\) takes the form
  \begin{equation}
    \label{eq:kth_derivative}
    \partial_\alpha f(\bm{x})  = (-2)^{k-1} {\rm Im}(i^{k} \ev{V(\bm{x}, {\bm{\alpha}})}{\psi}),
  \end{equation}
  for some operator \(V(\bm{x}, \alpha)\) which depends on both \(\alpha\)
  and \(\bm{x}\).
  Note that \(V\) is a product of terms which are either unitary, or from
  \(\{O_j\}\).
  Since $\|O_j\| \leq 1$ for all $j$, we have \(\|V\| \leq 1\), and therefore
  $|\partial_\alpha f(\bm{0})| \leq 2^{k-1}$ for all $k$ and $\alpha$.
  By setting \(c=2\), we satisfy the derivative conditions of
  \thm{gradient_algorithm}.

  To construct a probability oracle for \(f\)
  (see~\defi{probability_oracle_main_text}), we need a
  quantum circuit that encodes \(f(\bm{x})\) into the amplitudes of an
  ancilla.
  We construct such a circuit using the Hadamard test for the imaginary component
  of \(\bra{\psi}U(\bm{x})\ket{\psi}\)~\cite{Yu_Kitaev1995-zv, Aharonov2009-nk}.
  Let
  \begin{equation}
    \label{eq:param_circuit_def}
    F(\bm{x}) \defeq
    \big(H \otimes \idmat \big) \big( \controlled{U(\bm{x})} \big)
    \big(S^\dagger H \otimes U_\psi\big),
  \end{equation}
  where \(H\) denotes the Hadamard gate, \(\controlled{U(\bm{x})}\) the
  \(U(\bm{x})\) gate controlled on the first qubit, and \(S \coloneqq \ket{0}\!\!\bra{0} + i\ket{1}\!\!\bra{1}\) the phase gate.
  Applied to $\ket{0}\otimes \ket{\bm{0}}$, this circuit encodes \(f(\bm{x})\)
  in the amplitudes with respect to the computational basis states of the first qubit:
  \begin{align}
    \label{eq:amplitude_encode}
    F(\bm{x})\ket{0}\otimes\ket{\bm{0}} = & \sqrt{f(\bm{x})}\ket{1}\otimes \ket{\phi_1(\bm{x})} + \\ & \sqrt{1 - f(\bm{x})}\ket{0}\otimes \ket{\phi_0(\bm{x})}, \nonumber
  \end{align}
  for some normalized states $\ket{\phi_0(\bm{x})}$ and $\ket{\phi_1(\bm{x})}$ (see \app{details}{IV} for more details).
  Note that \(F(\bm{x})\) uses a single call to
  the oracle \(U_\psi\).

  \begin{figure}
    \centering
    \includegraphics[width=\linewidth]{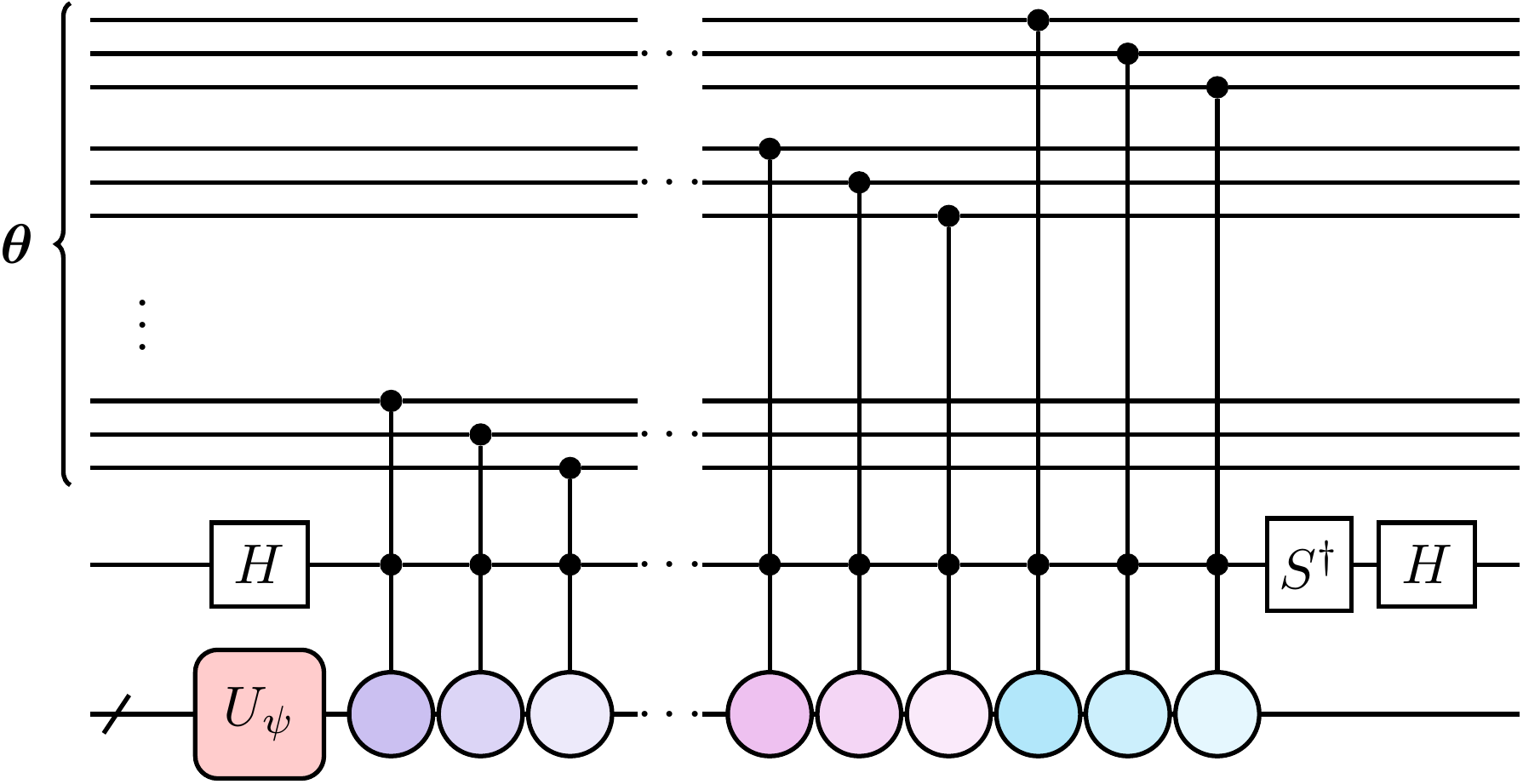}
    \caption{Schematic depiction of the quantum circuit for \(U_f\), the probability oracle for
      the function \(f(\bm{x})\) defined in \eq{U_f_def}.
      The top registers encode the ($n=3$ bit in this case) binary representations of \(x_1, x_2, \cdots, x_M\).
      The ancilla qubit whose amplitudes encodes \(f(\bm{x})\) (cf.~Eq.~\eqref{eq:param_circuit_def}) is indicated below the \(\bm{x}\) registers.
      The final line represents the \(N\)-qubit system register.
      The gates that act on the system register with colored circles represent the doubly-controlled time evolution by the various observables.
      Estimating the expectation values of the \(M\) observables \(\{O_j\}\) requires executing this circuit and its inverse \(\bigot{\sqrt{M}/\varepsilon}\) times.
    }
    \label{fig:circuit_diagram}
  \end{figure}

  All that remains is to add quantum controls to the rotations in
  \(F(\bm{x})\), so that
  $F(\bm{x})$ is controlled on a register encoding $\bm{x}$.
  Specifically, we consider the unitary
  \begin{align}
    \label{eq:U_f_def}
    U_f \defeq \sum_{\bm{k} \in G_n^M} \ketbra{\bm{k}} \otimes F(\bm{k} x_{\max}),
  \end{align}
  where \(G_n^{M}\) is a set of \(2^{nM}\) points distributed in an
  \(M\)-dimensional unit hypercube, with \(n = \bigo{\log(1/\varepsilon)}\), and
  \(x_{\max}\) is a rescaling factor.
  The values of \(x_{\max}\) and \(n\) are chosen to
  satisfy the requirements of the gradient algorithm (see
  \app{more_gradient_details}{IV}).
  Here, \(\ket{\bm{k}} = \ket{k_1}\dots \ket{k_M}\) for \(\bm{k} \in G_n^M\)
  denotes the basis state storing the binary representation of $\bm{k}$ in \(M\)
  \(n\)-qubit index registers.
  The controlled time evolution operator for each \(O_j\) can be implemented
  efficiently as a product of \(n\) controlled-$e^{-ix O_j}$ gates with
  exponentially spaced values of \(x\), each controlled on the appropriate
  qubit of the \(j\)th index register.
  We illustrate an example of such a \(U_f\) in \fig{circuit_diagram}.

  \(U_f\) is a probability oracle for the function
  \(f\), and each call to \(U_f\) involves a single call to the state preparation
  oracle \(U_\psi\).
  \thm{gradient_algorithm} then implies that with probability at least $2/3$, every component of the gradient of
  $f$, and hence all of the expectation values $\bra{\psi}O_j\ket{\psi}$, can be
  estimated to within an error \(\varepsilon\) using $\bigot{\sqrt{M}/\varepsilon}$ queries to $U_f$.
  The complexity in terms of the controlled time evolutions follows from
  multiplying the number of controlled time evolutions required for each query to
  \(U_f\), i.e.,~\(\bigo{\log(M/\varepsilon)}\) per observable, by the total number of queries, i.e.,~\(\bigot{\sqrt{M}/\varepsilon}\).
  As discussed in \app{more_gradient_details}{IV}, we have \(x_{\max} \in
  \bigo{1/\sqrt{M}}\) as a consequence of the details of the proof of
  \thm{gradient_algorithm} in \citen{Gilyen2017-gk}.
  This completes the proof of \thm{ub}.
\end{proof}
Furthermore (see \app{more_gradient_details}{IV}), the space complexity of the
gradient algorithm is the same as that of the probability oracle up to an
additive logarithmic factor~\footnote{To achieve this space complexity we
  actually need to compile the circuits for a logarithmic (in \(M\) and
  \(\varepsilon^{-1}\)) number of probability oracles across a series of hypercubes
  of varying sizes.
  Otherwise there would be additional multiplicative logarithmic factors in the
  space complexity.}.
Therefore, our algorithm uses \(\bigo{M\log(1/\varepsilon) + N}\) qubits.

\subsection*{Discussion}

In this letter, we considered the problem of simultaneously estimating the
expectation values of multiple observables with respect to a pure state
$\ket{\psi}$.
We presented an algorithm that uses
\(\widetilde{\mathcal{O}}(\sqrt{M}\varepsilon^{-1})\) applications of $U_\psi$ and
its inverse, where \(M\) denotes the number of observables and \(\varepsilon\) the
target error, and $U_\psi$ is a unitary that prepares $\ket{\psi}$.
We explained how a lower bound on a closely related problem posed
in~\citen{Van_Apeldoorn2021-sk} implies that, for algorithms given black-box
access to $U_\psi$, this query complexity is worst-case optimal up to logarithmic factors when \(\varepsilon \in (0, \frac{1}{3 \sqrt{M}})\).
In fact, our algorithm affirmatively resolves an open question
from~\citen{Van_Apeldoorn2021-sk} regarding the achievability of this bound for
the simultaneous estimation of classical random
variables~\footnote{Specifically, any matrix $A \in [-1,1]^M$ can be encoded in
  a known unitary $U_A$ along similar lines as in the proof of
  Corollary~\ref{cor:lower_bounds}; then $A\vb{p}$ can be estimated by applying
  our algorithm to the single-qubit $Z$ operators $\{Z_j\}$, with state
  preparation oracle $U_A(\mathbb{I}\otimes U_{\vb{p}})$.}.
These results imply that the optimal cost for expectation value
estimation can become exponentially worse with respect to \(M\) when one demands
a scaling that goes as \(\varepsilon^{-1}\) instead of \(\varepsilon^{-2}\).
Furthermore, the instances used in establishing our lower bounds involve a set of mutually commuting observables, implying that commutativity isn't necessarily helpful when one demands \(\varepsilon^{-1}\) scaling.

We presented a comparison with other approaches for
the estimation of expectation values in \tab{cost_comparison}, which we elaborate on in \app{prior_estimation_work}{I} and \app{applications}{II}.
For example, we find that our algorithm is capable of estimating each element of
the \(k\)-body fermionic reduced density matrix ($k$-RDM) of an \(N\)-mode system to within error \(\varepsilon\) using
\(\bigot{N^k/\varepsilon}\) state preparation queries.
This offers an unconditional asymptotic speedup compared to existing methods when \(\varepsilon = o(N^{-k/3})\).
This may be particularly useful in practical applications where we wish to
achieve a fixed error in extensive quantities by measuring the \(1\) or
\(2\)-RDM and summing \(\bigomega{N}\) elements.

Our gradient-based approach to estimating expectation values can be extended to other properties.
For example, consider the task of evaluating a collection of two-point dynamic
correlation functions.
These functions take the form
\begin{equation}
  \label{eq:dynamic_correlation_func}
  C_{A, B}(t) \defeq \ev{U(0,t) A^\dagger U(t, 0) B}{\psi},
\end{equation}
where \(A\) and \(B\) are some simple operators and \(U(t, t')\) is the time evolution
operator that maps the system from time \(t'\) to time \(t\).
These correlation functions are often directly accessible in experiment, as in the case of angle-resolved photoemission
spectroscopy~\cite{Damascelli2004-yk}, and are also central to hybrid
quantum-classical methods based on dynamical mean-field
theory~\cite{Bauer2016-hy, Georges1992-qf, Kotliar2006-rf}.
In \app{dynamic}{V}, we explain how a generalization of our approach can reduce the
number of state preparations required for estimating a collection of these
correlation functions.

Although we focused on quantifying the number of state preparation oracle queries, we also considered two other complexity
measures.
Our approach requires time evolution by each of the \(M\) observables.
The total duration of time evolution required scales as \(\bigot{M / \varepsilon}\).
We also need an additional \(\bigot{M \log(1/ \varepsilon)}\) qubits, although we can modify our approach to trade off between space and query complexities (see ~\app{time-space-trade-offs}{VI}).
When we are interested in simultaneously estimating \(O(N)\) expectation values, the asymptotic scaling of the space complexity is only logarithmically larger than that of storing the system itself.
This is the case in a variety of contexts, for example, in the evaluation of the momentum distribution~\cite{Meckel2008-mz}.
In other situations, the space overhead may be more substantial, though the capability of
modern simulation algorithms to use so-called ``dirty ancilla'' (temporarily
borrowing qubits in an arbitrary state) may offset this challenge in some
contexts~\cite{Lee2021-su, Von_Burg2021-gu, Low2018-dl}. 
As a concrete example, we consider the double-factorized simulation of the electronic structure Hamiltionian proposed in \citen{Von_Burg2021-gu}. Von Burg \emph{et al.} find that the time complexity of their simulation algorithm can be minimized by using \(\bigot{N^{3/2}}\) qubits for data-lookup. These same qubits could be used by our algorithm for expectation value estimation to parallelize the measurement of \(\bigot{N^{3/2}}\)  observables, offering a \(\bigot{N^{3/4}}\) asymptotic speedup without any additional qubit overhead.

Another potential reason for modifying our approach arises when the observables
of interest have different norms, or when the desired precision varies.
In
\app{eps_1_2_trade-offs}{VII}, we consider addressing this situation by measuring
certain observables using our strategy and measuring others using a
sampling-based method.
In \app{arbitrary_norms}{VIII}, we take a different approach,
and generalize Gily\'{e}n \emph{et al.'s} gradient estimation algorithm to
accommodate functions whose gradient components are not necessarily uniformly
bounded.
This allows us to simultaneously estimate the expectation values of
observables $\{O_j\}$ with arbitrary norms $\|O_j\|$ (possibly greater than $1$)
using $\widetilde{\mathcal{O}}(\sqrt{\sum_j \|O_j\|^2}/\varepsilon)$ queries.
By
rescaling the individual observables we can then also vary how precisely we
estimate each expectation value, thereby extending Theorem~\ref{thm:ub} to the
most general setting.

Our focus has been on the asymptotic scaling of our approach, but it will also be
desirable to understand the actual costs.
Performing a fault-tolerant resource estimate and a comparison against
other measurement strategies in the context of a practical application would be a useful
line of future work.
It is possible that our approach could be modified to obtain a further speedup
by taking advantage of the structure of the states and/or observables for
particular problems of interest.
Another potentially fruitful direction would be to explore extensions of the
gradient algorithm to yield quantum algorithms for the Hessian or even
higher-order derivatives.

Extracting useful information from a quantum computation, especially a quantum
simulation, is a bottleneck for many applications.
This is especially true in fields such as quantum chemistry and materials
science, where it may be necessary to couple high-level quantum calculations
with coarser approximations at other length scales in order to describe
macroscopic physical phenomena.
We expect that our gradient-based approach to the estimation of expectation values will be a useful tool and a
starting point for related approaches to other problems.

\subsection*{Acknowledgements}

The authors thank Bryan O'Gorman, Yuan Su, and Joonho Lee for helpful discussions and various referees for their constructive input.
NW worked on this project under a research grant from Google Quantum AI and was
also supported by the U.S.~Department of Energy, Office of Science, National
Quantum Information Science Research Centers, Co-Design Center for Quantum
Advantage under contract number DE-SC0012704.
Some discussion and collaboration on this project occurred while using facilities at the Kavli Institute for Theoretical Physics, supported in part by the National Science Foundation under Grant No. NSF PHY-1748958.

\nocite{Lin2020-iy,wan2020fast,Verteletskyi2020-ei, Huggins2021-vu,Chen2021-cq,
Hadfield2020-lg, Brandao2017-gl, Huang2021-prm, Davidson2012-ci, bonet2020nearly, Somma2002-wq, NCbook}

%

\widetext
\appendix
\setcounter{secnumdepth}{2}
\section{Prior work on expectation value estimation}
\label{app:prior_estimation_work}

This letter focuses on the task of estimating the expectation value of multiple
observables with respect to a pure state $\ket{\psi}$.
Motivated by settings where the cost of state preparation is the dominant
factor, we mainly quantify the resources required in an oracle model where we
count the number of calls to the state preparation unitary and its inverse.
To provide concrete motivation for this cost model and for the task in general,
consider the example where our state of interest is the unknown ground state of
some second-quantized electronic structure Hamiltonian under the Jordan-Wigner
transformation.
In this case, the state preparation step is expected to be tractable under
certain assumptions but relatively expensive, even using modern methods (e.g.,
by applying the ground state preparation algorithms of
\citen{Lin2020-iy,wan2020fast} in conjunction with state-of-the-art techniques
for block-encoding the Hamiltonian~\cite{Lee2021-su,Von_Burg2021-gu}).
At the same time, the observables of interest may be especially simple (e.g.,
the elements of a fermionic reduced density matrix).
We discuss the situation where the cost of state preparation does not necessarily
dominate, and the possible trade-offs available in the context of our approach,
in \app{time-space-trade-offs}.

Let \(U_\psi\) denote the unitary which prepares \(\ket{\psi}\) from the
\(\ket{0}\) state, and let \(\{O_j\}\) be a collection of \(M\) Hermitian
operators. For the sake of simplifying the comparison with existing approaches,
we make the additional assumption in this section that the \(O_j\) are also
unitary, though this requirement could be relaxed by using techniques based on
block-encodings~\cite{Rall2020-in}. As in the main text, our goal is to minimize
the number of calls to $U_\psi$ and $U_\psi^\dagger$ required to obtain
estimates $\widetilde{o_j}$ of the $M$ expectation values
$\bra{\psi}O_j\ket{\psi}$ such that
\begin{equation}
  \label{eq:ev_goal}
  |\widetilde{o_j} - \ev{O_j}{\psi}| \leq \varepsilon
\end{equation}
for all $j \in \{1,\dots, M\}$ with probability at least $2/3$.
We note that some of the methods we compare against are usable under a weaker
input model, where the
algorithm is merely provided copies of \(\ket{\psi}\) rather than access to a state preparation oracle.

A straightforward approach is to repeatedly prepare the state \(\ket{\psi}\) and
simultaneously measure as many of the operators as possible on each copy.
To this end, consider dividing the \(M\) operators into
\(G\) groups of mutually commuting terms.
Then
\begin{equation}
  \label{eq:standard_approach_cost}
  C = \mathcal{O}\left(\frac{G \log M}{\varepsilon^2}\right)
\end{equation}
calls to \(U_\psi\) suffice to estimate the $M$ expectation values.
This can be accomplished by simultaneously diagonalizing the operators within
each group and measuring in the common eigenbasis when this is tractable, or by
performing phase estimation on the operators within each group using the same
copy of $\ket{\psi}$.
The outcomes are then averaged over repeated iterations.
One key advantage of this approach is that, although it has poor scaling with
the target error $\varepsilon$, it does not scale polynomially with $M$, and
instead scales with $G$, which could be considerably smaller than $M$.
In practice, optimally grouping and sampling the observables may be challenging
and can introduce substantial overheads not captured by the query complexity
(e.g., the classical cost of finding optimal
groupings~\cite{Verteletskyi2020-ei, Huggins2021-vu} and of simultaneous
diagonalization, the quantum gate complexity of implementing basis change
unitaries corresponding to the common eigenbasis, or of phase estimation).

Alternatively, using a strategy based on amplitude
estimation~\cite{Brassard2000-oi}, we can estimate expectation values with a
scaling proportional to \(\varepsilon^{-1}\)~\cite{Knill2007-ci, Rall2020-in}.
Amplitude estimation, as originally implemented in \citen{Knill2007-ci}, allows for the estimation of the expectation value
$\langle\psi|U|\psi\rangle$
of a unitary operator $U$.
estimation algorithm of~\citen{Knill2007-ci}.
The amplitude estimation algorithm works by performing phase estimation
on the Szegedy walk operator,
\begin{equation}
  S=-RURU^{\dag},\hspace{0.5cm}R=1-2|\psi\rangle\langle\psi|.
\end{equation}
One can verify that the operator $S$ is diagonal in the subspace spanned by $|\psi\rangle$ and $U|\psi\rangle$, with eigenvalues $e^{\pm i x}$ for $x = \cos^{-1}(2|\langle\psi|U|\psi\rangle|^2-1)$.
Phase estimation on \(S\) therefore allows for the determination of \(|\ev{\psi}{U}|^2\) up to an accuracy \(\eps\) with a cost that scales as \(1/\eps\). Information about the phase may be regained by repeating this for different variants of \(U\) controlled by an ancilla qubit.
One can generalize this approach to the estimation of the expectation value of an arbitrary observable \(O\) by providing a \(U\) that block encodes \(O\)~\cite{Rall2020-in}.
Unfortunately, this algorithm does not generalize in a straightforward way to
the task of estimating multiple expectation values, even when the operators
commute (beyond the strategy of treating each one separately).
As a consequence, estimating all \(M\) expectation values using amplitude
estimation requires
\begin{equation}
  \label{eq:serial_amplitude_estimation_cost}
  C = \widetilde{\mathcal{O}}\left({\frac{M}{\varepsilon}}\right)
\end{equation}
calls to \(U_\psi\) and \(U_\psi^\dagger\).
While this leads to an asymptotic advantage over naive sampling in some
settings, it fails to do so in cases where $G/\varepsilon = o(M)$.

These two well-known strategies are complemented by the more recent body of work that began with Aaronson's definition of the ``shadow tomography'' problem,
the problem of estimating the expectation value of many two-outcome measurements
given multiple copies of some input state~\cite{Aaronson2020-ei}.
\citen{Aaronson2020-ei} proposes a computationally expensive protocol for this
task that achieves scaling logarithmic in the number of two-outcome measurements $M$ and
proportional to \(\varepsilon^{-4}\), up to logarithmic factors.
\citen{Huang2020-kn} put forward an alternative protocol based on randomized
measurement that also scales logarithmically in $M$, and improves the scaling
with $\varepsilon$ from \(\bigot{\varepsilon^{-4}}\) to \(\bigot{\varepsilon^{-2}}\) at
the expense of limiting the types of observables that can be treated
efficiently (i.e., without introducing a scaling polynomial in the Hilbert space
dimension).
A series of additional works have offered improvements and variations on both
the randomized measurement approach~\cite{Zhao2021-fv, Chen2021-cq,
  Hadfield2020-lg}, and the more general approaches based on gentle
measurements~\cite{Brandao2017-gl,Van_Apeldoorn2019-xl, Huang2021-pr}.

In this work, we are primarily concerned with the high-precision regime, and aim
to achieve \(\bigot{\varepsilon^{-1}}\) scaling.
It is natural to ask whether it is possible to simultaneously achieve an
asymptotic cost that is logarithmic in the number of observables $M$ and
linear (up to logarithmic factors) in \(\varepsilon^{-1}\) using our input
model, where the input state is unknown and accessed via a black-box state
preparation unitary.
In Corollary III in the main text, we point out that recent results preclude this, showing that the
desired \(\varepsilon^{-1}\) scaling in the precision necessarily comes at the cost
of a square root dependence on $M$ for certain collections of observables.
On the other hand, \citen{Van_Apeldoorn2021-sk} provides an example of a collection
of operators---namely, projectors onto orthogonal states---where this scaling is achievable.

\section{Applications}
\label{app:applications}

In this appendix, we apply our method to three illustrative cases and consider
the potential for asymptotic speedups over alternative approaches. As in the main text, we focus on quantifying the cost with respect to the number of state preparation oracle queries. We note that some of the approaches we compare against (namely, methods based on sampling and shadow tomography) are usable under a weaker access model, where copies of the state are provided rather than queries to the state preparation oracle.

A major application of these ideas is the measurement of the fermionic $k$-body
reduced density matrices ($k$-RDMs) of a particular pure state.
The \(k\)-RDM of a pure state \(\ket{\chi}\) with support on \(N\) fermionic modes is
a tensor specified by the \(N^{2k}\) ``matrix elements'' of the form $\bra{\chi}
  a^\dagger_{p_1}\cdots a^{\dagger}_{p_k} a_{p_{k+1}} \cdots
  a_{p_{2k}}\ket{\chi}$, where the \(2k\) indices \(p_j\) take values ranging over
the \(N\) modes.
The $1$- and $2$- RDMs are particularly important, being sufficient to determine
the expected energy of a state and many other properties of interest~\cite{Davidson2012-ci}.

The $\mathcal{O}(N^{2k})$ terms in
the \(k\)-RDM can be divided into $\Omega(N^{k})$ groups of $\mathcal{O}(N^{k})$
mutually commuting operators~\cite{bonet2020nearly}. This allows for each of the terms to be estimated to within a precision \(\varepsilon\) by a simple sampling procedure with a query complexity of $\bigot{N^{k}/\varepsilon^{2}}$~\cite{Zhao2021-fv}.
The asymptotic scaling with $\varepsilon$ can be quadratically improved
(see~\app{prior_estimation_work}) by applying amplitude estimation to learn each
component individually with the requisite error, leading to a query complexity of
$\bigot{N^{2k}/\varepsilon}$.
Some shadow tomography protocols scale better with respect to \(N\) and \(k\) than either of these approaches, at the expense of scaling
with \(\varepsilon^{-4}\)~\cite{Huang2021-pr, Aaronson2020-ei, Brandao2017-gl,
  Van_Apeldoorn2019-xl}.
In particular, the technique described in \citen{Huang2021-pr} can estimate the
\(k\)-RDM at a cost that scales as \(\bigo{k \log(N) / \varepsilon^4}\) by estimating the expectation values of all degree \(2k\) majorana operators under the Jordan-Wigner transformation to within a precision \(\varepsilon\) and using these values to reconstruct the \(k\)-RDM.
Our gradient-based algorithm's scaling of \(\bigot{N^k/\varepsilon}\) for this
application follows directly from Theorem 1 in the main text.
In terms of the number of state preparation oracle queries, our method thus
strictly improves upon sampling and amplitude estimation for this application,
and provides an asymptotic advantage over all prior approaches for learning
the \(k\)-RDM when \(\varepsilon \in o(N^{-k/3})\).
This might naturally occur when we are interested in obtaining estimates of some
extensive quantities to within a fixed precision by summing estimates of \(O(N)\) local
observables.

For the case where we wish to apply our ideas to compute the expectation values
of $M$ mutually commuting observables with respect to a state $\ket{\psi}$, the
potential benefit still exists, but the trade-offs are less favorable. In this
case, commutation allows every expectation value to be measured on a single copy
of $\ket{\psi}$. Sampling then yields $\bigot{\log (M)/\varepsilon^2}$ scaling,
where the logarithmic dependence on \(M\) arises from the application of
concentration inequalities together with the union bound to guarantee that all
of the \(M\) expectation values are estimated to within \(\varepsilon\)
simultaneously (with some constant probability of success). This is in contrast
to the $\bigot{\sqrt{M}/\varepsilon}$ scaling of Theorem 1 in the main text. This does not
contradict the optimality of our approach in the high-precision regime, due to
the requirement that $\varepsilon < 1/(3\sqrt{M})$
in Theorem 2 and Corollary 3 in the main text. Inside this region of
applicability, sampling has at best the same scaling as our algorithm. The
trade-offs between sampling and quantum gradient approaches, as well as
algorithms that blend the two, are discussed in~\app{time-space-trade-offs}.

Another simple case to consider is one where we wish to measure \(M\)
observables that all fail to commute.
Then it is clear that our approach offers an unconditional speedup when compared
with approaches based on sampling and amplitude estimation (in terms of the
number of state preparation oracle queries).
It is still possible in this case to obtain a better scaling with respect to \(M\)
by employing shadow
tomography~\cite{Aaronson2020-ei, Brandao2017-gl, Van_Apeldoorn2019-xl,
  Huang2021-pr}.
We provided a brief discussion of these approaches in
\app{prior_estimation_work}.
To summarize, these strategies require a number of state preparation oracle
calls (actually, copies of the state) that scale logarithmically (or
poly-logarithmically) with \(M\) at the cost of scaling with \(\varepsilon^{-4}\).
When \(\varepsilon = o(M^{-1/6})\) our approach has a favorable asymptotic
scaling in terms of the state preparation oracle query complexity.

We note that these shadow tomography protocols are limited in other ways that
may ultimately lead to a gate complexity advantage for our method for particular
applications, even when the query complexity advantage is absent.
Some such protocols, such as the one in \citen{Aaronson2020-ei}, have gate
complexities that scale exponentially in one or more relevant parameters.
Others, such as the one in \citen{Huang2021-pr}, are computationally efficient
but only apply to limited sets of operators (Pauli observables in the case of~\citen{Huang2021-pr}).
The proposal of \citen{Brandao2017-gl} avoids both of these obstacles but
returns a representation of the state in terms of a quantum circuit that must
itself be prepared and measured to obtain the expectation values of interest.
In some cases, it might be fruitful to apply our measurement techniques to the
state whose preparation circuit is learned by this latter proposal.

\section{Further background on Gily\'{e}n et al.'s quantum algorithm for gradient estimation}
\label{app:more_gradient_details}
In this appendix, we restate a few of the important definitions used in the main
theorem which summarizes the performance of the quantum algorithm for the
gradient (Theorem 4 in the main text of this work, Theorem 25 in
\citen{Gilyen2017-gk}).
We also point out some of the details of the gradient algorithm relevant to our
consideration of costs beyond the phase oracle complexity that are not included
in the main statement of the theorem.
We refer the interested reader to \citen{Gilyen2017-gk} for a rigorous analysis
and further information about the implementation.

In Theorem 4 of the main text, we referred to the probability oracle
access model for \(f\).
We recall this definition here, as well as the definitions for the phase oracle access model from \citen{Gilyen2017-gk} and the binary
oracle access model used in Jordan's original gradient
algorithm~\cite{Jordan2005-hs}.
\begin{definition}[Probability oracle]
  \label{def:probability_oracle}
  Consider a function \(f: \mathbb{R}^M \rightarrow [0, 1]\).
  A probability oracle \(U_f\) for \(f\) is a unitary operator that acts
  as
  \begin{equation} \label{eq:U_f}
    U_f: \ket{\bm{x}}\ket{\bm{0}} \mapsto
    \ket{\bm{x}}\left(\sqrt{f(\bm{x})}  \ket{1} \ket{\phi_1(\bm{x})} + \sqrt{1 - f(\bm{x})}\ket{0}\ket{\phi_0(\bm{x})}\right),
  \end{equation}
  where \(\ket{\bm{x}}\) denotes a discretization of the variable \(\bm{x}\)
  encoded into a register of qubits, \(\ket{\bm{0}}\) denotes the all-zeros state of a register of ancilla qubits, and \(\ket{\phi_0(\bm{x})}\) and \(\ket{\phi_1(\bm{x})}\) are
  arbitrary quantum states.
\end{definition}
\begin{definition}[Phase oracle]
  \label{def:phase_oracle}
  Consider a function \(f: \mathbb{R}^M \rightarrow \mathbb{R}\).
  A phase oracle \(A_f\) for \(f\) is a unitary operator that acts as
  \begin{equation}
    A_f: \ket{\bm{x}}\ket{\bm{0}} \mapsto e^{i f(\bm{x})} \ket{\bm{x}}\ket{\bm{0}},
  \end{equation}
  where \(\ket{\bm{x}}\) denotes a discretization of the variable \(\bm{x}\)
  encoded into a register of qubits and \(\ket{\bm{0}}\) denotes the all-zeros state of a register of ancilla qubits.
\end{definition}
\begin{definition}[Binary oracle]
  \label{def:binary_oracle}
  Consider a function \(f: \mathbb{R}^M \rightarrow \mathbb{R}\).
  For some precision parameter \(\varepsilon>0 \), an
  \(\varepsilon\)-accurate binary oracle \(B_f\) for \(f(\bm{x})\) is a unitary
  operator that acts as
  \begin{equation}
    B_f: \ket{\bm{x}}\ket{\bm{0}} \mapsto \ket{\bm{x}}|\widetilde{f}(\bm{x})\rangle,
  \end{equation}
  where \(\ket{\bm{x}}\) denotes a discretization of the variable \(\bm{x}\)
  encoded into a register of qubits, \(\ket{\bm{0}}\) denotes the all-zeros state of a register of ancilla qubits, and \(\widetilde{f}(\bm{x})\) denotes a fixed-point
  binary number such that \(|f(\bm{x}) -\widetilde{f}(\bm{x})| \leq \varepsilon\).
\end{definition}

While we referred to the probability oracle access model in our Theorem 4 (in
the main text), the original Theorem~25 of
\citen{Gilyen2017-gk} described the gradient algorithm purely in terms of the
phase oracle access model. However, they also show how we can efficiently obtain a
probability oracle from a phase oracle.
\begin{theorem}[Theorem 14, \citen{Gilyen2017-gk}]
  \label{thm:oracle_conversion}
  Let \(U_f\) be a probability oracle for a function \(f(\bm{x})\).
  For any \(\varepsilon \in (0, \frac{1}{3})\), we can implement an
  \(\varepsilon\)-approximate phase oracle \(\widetilde{A}_f\) such that
  \begin{equation}
    \label{eq:approximate_phase_oracle}
    \| \widetilde{A}_f\ket{\psi}\ket{\bm{0}} - A_f\ket{\psi}\ket{\bm{0}}  \| \leq \varepsilon
  \end{equation}
  for all input states \(\ket{\psi}\), where \(A_f\) denotes an exact phase oracle
  for \(f\).
  This implementation uses \(\bigo{\log({1}/{\varepsilon})}\) invocations
  of the probability oracle \(U_f\) (or its inverse) and \(\bigo{\log
    \log({1}/{\varepsilon})}\) additional ancilla qubits beyond those required by
  \(U_f\).
\end{theorem}
Therefore, we can regard calls to a phase oracle for \(f\) as equivalent to
calls to a probability oracle for \(f\) up to logarithmic factors.
More technically, the actual implementation of the gradient algorithm requires
the use of a modified phase oracle known as a fractional query phase oracle.
A fractional query phase oracle is defined in the same way as the phase oracle
of \defi{phase_oracle}, except that it has an additional parameter \(s \in [-1,
  1]\) which rescales the argument of the exponential.
Fortunately, \citen{Gilyen2017-gk} explains how a fractional phase oracle can
be naturally arrived at from a probability oracle in a theorem closely related
to \thm{oracle_conversion} (essentially by applying a rotation to rescale the amplitude $\sqrt{f(\bm{x})}$ of $\ket{1}$ in Eq.~\eqref{eq:U_f} to \(\sqrt{sf(\bm{x})}\), before converting to a phase oracle).

The statement of Theorem~1 in the main text gives the cost of the gradient estimation algorithm in terms of
the number of calls to a probability oracle. Similarly, for our expectation value estimation algorithm, we mainly focus on quantifying
the cost in terms of oracle complexity (in particular, the number of queries to the state
preparation oracle).
However, we also wish to describe some of the secondary costs that we encounter
in our algorithm, so it is useful to note a few additional
details about the gradient estimation algorithm.

A secondary cost we might consider is the amount of time evolution
required for each observable.
From the proof of Theorem 25 in \citen{Gilyen2017-gk}, we know that the phase
oracle is queried at uniform superpositions of points within a series of
\(M\)-dimensional boxes and that the largest such box has a side length of
\begin{equation}
  x_{\max} \defeq rm,
\end{equation}
where \(r\) is defined implicitly by the equation
\begin{equation}
  \label{eq:r_def}
  r^{-1} \defeq 9cm \sqrt{M} (81 \cdot 8 \cdot 42\pi c m \sqrt{M}/\varepsilon)^{1/(2m)},
\end{equation}
and \(m\) is a positive integer,
\begin{equation}
  m \defeq \big \lceil \log(c \sqrt{M} / \varepsilon) \big \rceil.
\end{equation}
Here, \(c\) is the same fixed constant as that introduced in Theorem 4 in the main text for bounding the partial derivatives of \(f\).
As a consequence of these expressions, we have that the largest side length
shrinks as \(M\) increases.
Specifically, \(x_{\max} = \bigo{1/\sqrt{M}}\).
Calls to the phase oracle are generated using calls to the probability oracle
and its inverse over the same input parameters (see the proof of Theorem 14 in
\citen{Gilyen2017-gk}, stated above as \thm{oracle_conversion} for convenience)
and the box size for the probability oracle directly determines the amount of
time evolution by each observable.
Therefore, for each probability oracle query, we require at most
\(\bigo{1/\sqrt{M}}\) units of time evolution by each observable.

The other substantial secondary cost to consider is the space complexity.
The probability oracle \(U_f\) for our expectation value algorithm can be
implemented using \(N + 1 + Mn\) qubits, where \(n =
\bigo{\log(1/\varepsilon)}\).
The additive cost of \(N + 1\) comes from the system register and the ancilla
for the Hadamard test.
The \(Mn\) terms is due to the need for \(M\) \(n\)-bit ancilla registers to
prepare a superposition over states indexing the \(2^{nM}\) points in the
hypercube \(G_n^M\).
To be specific, this hypercube is composed of the Cartesian product of \(M\)
copies of the set \(G_n\), defined as
\begin{equation}
  \label{eq:gndef}
  G_n \defeq \left\{ j/2^n - 1/2 + 1/2^{n+1} : j \in \left\{ 0, 1,
  ..., 2^{n} - 1\right\} \right\}.
\end{equation}
The logarithmic scaling of \(n\) with $1/\varepsilon$ comes from the precision
requirements of the gradient algorithm (see the definition of \(n\) in Theorem
21 in \citen{Gilyen2017-gk} and note that the factors of \(\ell \in \{-m, -m +
1, ..., m\}\) that appear in the argument of the oracle in Theorem 25 can be
accounted for by compiling a family of \(2m + 1\) related phase/probability
oracles on grids of varying size).
The conversion from a probability oracle to an \(\varepsilon'\)-approximate phase oracle requires only \(\bigo{\log\log(1/\varepsilon')}\)
additional ancilla.
The gradient algorithm nominally requires storing \(\bigo{\log(M/\delta)}\)
copies of each value and performing a coherent median finding step, but this can
be performed classically for our purposes, eliminating the need for an additional
multiplicative factor of \(\bigo{\log M}\) in the number of qubits.
We therefore have that the space complexity for our estimation algorithm is \(N
+ 1 + Mn + \bigo{\log\log(Q/\varepsilon)}\), where \(n =
\bigo{\log(1/\varepsilon)}\) (see above) and \(Q = \bigot{\sqrt{M}/\varepsilon}\) is
the number of queries we make to the phase oracle.

The oracle, gate, and qubit complexities of our algorithm for estimating the expectation values of a general collection of observables (with arbitrary norms) are analyzed in \app{arbitrary_norms}; see SI Theorem~\ref{thm:ub_general} for an explicit statement.

\section{Additional details regarding gradient-based expectation value estimation}
\label{app:details}

In the main text, we claimed that the circuit \(F(\bm{x})\), defined in Equation 11 as
\begin{equation}
  \label{eq:param_circuit_def_app}
  F(\bm{x}) \defeq
  \big(H \otimes \idmat \big) \big( \controlled{U(\bm{x})} \big)
  \big(S^\dagger H \otimes U_\psi\big),
\end{equation}
encodes the function
\begin{equation}
  \label{eq:encodes_the_function}
  f(\bm{x}) = -\frac{1}{2}{\rm Im}\Bigg[\ev{\prod_{j=1}^M e^{-2 i x_j O_j}}{\psi}\Bigg] + \frac{1}{2}
\end{equation}
in the following way:
\begin{equation}
  \label{eq:amplitude_encode_app}
  F(\bm{x})\ket{0}\otimes\ket{\bm{0}} = \sqrt{f(\bm{x})}\ket{1}\otimes \ket{\phi_1(\bm{x})} + \sqrt{1 - f(\bm{x})}\ket{0}\otimes \ket{\phi_0(\bm{x})},
\end{equation}
where \(\ket{\phi_0(\bm{x})}\) and \(\ket{\phi_1(\bm{x})}\) are some arbitrary normalized quantum states.
While checking this identity directly is burdensome, it is easy to verify its correctness by observing that \(F(\bm{x})\) is the circuit that performs the the Hadamard test for the imaginary component of \(\ev{\prod_{j=1}^M e^{-2 i x_j O_j}}{\psi}\)~\cite{Yu_Kitaev1995-zv, Aharonov2009-nk}. That is, the expectation value of the Pauli \(Z\) operator on the ancilla qubit with respect to the state \(F(\bm{x})\ket{0}\otimes\ket{\bm{0}}\) is, by construction, equal to \({\rm Im}\Big[\ev{\prod_{j=1}^M e^{-2 i x_j O_j}}{\psi}\Big]\).

Let \(c_0\) and \(c_1\) be defined implicitly by the following expression,
\begin{equation}
  F(\bm{x})\ket{0}\otimes\ket{\bm{0}} = c_1\ket{1}\otimes \ket{\phi_1(\bm{x})} + c_0\ket{0}\otimes \ket{\phi_0(\bm{x})}.
\end{equation}
Our observation that the circuit from \eq{param_circuit_def_app} performs the Hadamard test for the imaginary component of \(\ev{\prod_{j=1}^M e^{-2 i x_j O_j}}{\psi}\) implies that
\begin{equation}
  |c_0|^2 - |c_1|^2 = {\rm Im}\Big[\ev{\prod_{j=1}^M e^{-2 i x_j O_j}}{\psi}\Big], \;\;\;\; |c_0|^2 + |c_1|^2 = 1.
\end{equation}
It is then straightforward to verify that \eq{encodes_the_function} and \eq{amplitude_encode_app} follow from \eq{param_circuit_def_app} by solving for \(c_0\) and \(c_1\) (absorbing the phases into the definitions of the arbitrary states \(\ket{\phi_0(\bm{x})}\) and \(\ket{\phi_1(\bm{x})}\)).

\section{Estimating dynamic correlation functions}
\label{app:dynamic}

In the main text, we considered the problem of the expectation values of multiple observables
with respect to a given pure state.
Our gradient-based approach can also be naturally applied to estimate other properties of
interest.
As a concrete example, we can use it to evaluate a
collection of two-point dynamic correlation functions.
Specifically, we consider functions of the form
\begin{equation}
  \label{eq:dynamic_correlation_func_app}
  C_{A, B}(t) \defeq \ev{U(0,t) A^\dagger U(t, 0) B}{\psi},
\end{equation}
where \(A\) and \(B\) are some operators of interest, and \(U(t, t')\) is
the time-evolution operator that maps a state at time \(t'\) to a state a time
\(t\).
Quantum algorithms for measuring individual dynamic correlation functions are
well known~\cite{Somma2002-wq, Bauer2016-hy}.
These quantities are useful when comparing with the direct outcomes of
spectroscopic experiments~\cite{Damascelli2004-yk}, and in the design of hybrid
quantum-classical methods based on dynamical mean field
theory~\cite{Bauer2016-hy, Georges1992-qf, Kotliar2006-rf}.

In order to proceed, we introduce some assumptions and notation.
To simplify the presentation and comparison with other algorithms, we assume
that \(A_1,\dots, A_M\) and \(B\) are operators that are both Hermitian and unitary, although the unitarity
condition is not required by our gradient-based approach.
More general \(A_j\)'s and \(B\) can be treated with a variety of methods, the most
simple of which is decomposing them into a linear combination of suitable
operators.
Let \(\{C_{A_1, B}(t_1), C_{A_2, B}(t_2), \cdots, C_{A_M, B}(t_M)\}\)
denote a collection of correlation functions we would like to evaluate.
Without loss of generality, we can assume that the time points are indexed in
nondecreasing order (\(t_1 \leq t_2 \leq \cdots \leq t_M\)).

Before describing our improved approach to estimating these quantities, we
briefly consider the cost of estimating them using standard amplitude
estimation-based techniques.
We quantify the cost in terms of two resources, calls to a unitary state
preparation oracle (and its inverse) for \(\ket{\psi}\), and the total amount of
time evolution under the system Hamiltonian.
We assume that the costs of applying \(A_j\) (for all \(j\)) and \(B\) to some
state, as well as performing \(\bigo{1}\) units of time evolution by these
operators, are all negligible.
This is particularly reasonable for dynamic correlation functions, where \(B\)
and the \(A_j\)'s are frequently some simple local operators.
As with the case of expectation value estimation, we are interested in the cost
(up to logarithmic factors) required to estimate each quantity to within some
additive error \(\varepsilon\) with a success probability of at least \(2/3\).
The naive approach is to use amplitude estimation to evaluate each quantity
separately, resulting in a resource cost of
\begin{equation}
  C_\psi = \widetilde{\mathcal{O}}\left(\frac{M}{\varepsilon}\right)
\end{equation}
calls to the state preparation oracle \(U_\psi\) and it's inverse, along with
\begin{equation}
  C_H = \mathcal{O}\left({\frac{\log M}{\varepsilon}}\sum_{j=1}^{M} t_j\right)
\end{equation}
units of time evolution under the system Hamiltonian.

Our alternative approach achieves an unconditional advantage in the number of
state preparation calls and may achieve an advantage with respect to the total
duration of time evolution, depending on the choice of the time points
\(t_j\).
For convenience, we define \(t_0 \defeq 0\).
We proceed as in the expectation value estimation case, constructing a
parameterized unitary for use with the Hadamard test,
\begin{equation}
  U(\bm{x}) \defeq \Bigg( \prod_{j=1}^M U(t_{j-1}, t_j) e^{- 2i x_j A_j} \Bigg) U(t_M, t_0) B.
\end{equation}
Taking the derivative of \(U(\bm{x})\) with respect to \(x_\ell\) and evaluating the resulting expression at
\(\bm{x} = \bm{0}\), we have
\begin{equation}
  \frac{\partial U(\bm{x})}{\partial x_\ell} \Bigg|_{\bm{x} = \bm{0}} =
  -2i \,U(t_{0}, t_\ell) A_\ell U(t_\ell, t_0) B.
\end{equation}
From this we can see that the \(M\) different partial derivatives of
\(U(\bm{x})\) with respect to \(x_i\) are the \(M\) different unitary
operators whose matrix elements we would like to estimate.
Just as we did in Equation 11 in the main text, we can apply a Hadamard test to
\(U(\bm{x})\).
We can then add quantum controls to the rotation angles \(\bm{x}\) to
construct a probability oracle for a function whose gradient yields the real
parts of the matrix elements of interest.
We could likewise use the Hadamard test for the real component of
\(U(\bm{x})\) to obtain the imaginary components of the matrix elements of
interest.
It is simple to show that the resulting functions satisfy the technical
conditions of Theorem 4 from the main text.
We can then apply the quantum algorithm for the gradient.

We analyze the asymptotic scaling of this approach.
Each application of the Hadamard test circuit for \(U(\bm{x})\) requires a
single call to the state preparation oracle $U_\psi$ and its inverse, plus \(2
\sum_{j=1}^M t_j \) units of time evolution.
We are interested in estimating \(M\) different quantities to within a precision
\(\varepsilon\), and so we require \(\bigot{{\sqrt{M}}/{\varepsilon}}\) calls to
our probability oracle.
Therefore, we require
\begin{equation}
  C_\psi = \widetilde{\mathcal{O}}\left(\frac{\sqrt{M}}{\varepsilon}\right)
\end{equation}
calls to $U_\psi$ and $U_\psi^\dagger$, along with
\begin{equation}
  C_H = \widetilde{\mathcal{O}}\left({\frac{\sqrt{M}t_M}{\varepsilon}}\right)
\end{equation}
units of time evolution under the system Hamiltonian.
Regardless of the chosen time points, the scaling in the number of state
preparation oracle calls is a factor of \(\sqrt{M}\) smaller than for a scheme
based on amplitude estimation.
If \(t_M = o({\sum_{j=1}^M t_j}/{\sqrt{M}})\), then using our
approach also scales more favorably in terms of the total time evolution.
For example, consider the case where the time points are evenly spaced in
increments of \(\Delta\).
Then our approach requires \(\bigot{{M^{1.5} \Delta}/{\varepsilon}}\) units of
time evolution, whereas the approach based on estimating each value  independently using amplitude estimation
requires \(\bigot{{M^2 \Delta^2}/{\varepsilon}}\) units.

\section{Trading off between state preparation, time evolution, and space}
\label{app:time-space-trade-offs}
We have proposed a strategy for estimating the expectation values of \(M\)
observables with respect to an \(N\)-qubit state \(\ket{\psi}\).
Neglecting logarithmic factors, our approach requires
\(\bigot{{\sqrt{M}}/{\varepsilon}}\) sequential calls to the state preparation
oracle for \(\ket{\psi}\) and \(\bigo{M\log(M/\varepsilon) + N}\) qubits to
estimate each expectation value to within error $\varepsilon$ with probability at
least \(2/3\).
It also requires implementing controlled time evolution gates for each
observable $O_j$; these are of the form $e^{-ix O_j}$ for various times
$x$, and the total time is $\widetilde{O}(1/\varepsilon)$ for each $O_j$.
Treating all of the observables as equivalent, we can say that the algorithm
requires a total of \(\bigot{M/\varepsilon}\) units of (controlled) time evolution
overall.

Note that applying our algorithm (or a strategy based on amplitude estimation)
to each observable separately would require
$\widetilde{\mathcal{O}}(M/\varepsilon)$ state preparation queries, but only
$\mathcal{O}(N + \log(1/\varepsilon))$ qubits would be needed.
In some contexts, we expect that the dominant cost will be that of implementing
of the state preparation oracle, in which case it would be advantageous to incur
the additional qubit overhead of $\widetilde{\mathcal{O}}(M)$, to reduce the
number of oracle queries by $\bigo{\sqrt{M}}$.
However, if space is also a limiting factor, we can interpolate between these
two extremes by dividing the observables into \(g\) groups of $M/C$ observables
each and applying our algorithm to each group separately.
Then, the total number of oracle calls required is \(\bigot{{\sqrt{gM
    }}/{\varepsilon}}\), while the number of qubits is reduced to \(\bigot{N +
{M}/{g}}\).
For \(g = \mathcal{O}(M)\), we recover the same asymptotic scaling (up to
logarithmic factors) in the query and the space complexity as the approach based
on applying our algorithm (or amplitude estimation) to each observable
separately.
The complexity with respect to the number of units of time evolution remains the
same regardless of the value of \(g\) (up to logarithmic factors).

\section{Grouping observables to trade off between gradient-based estimation and sampling}
\label{app:eps_1_2_trade-offs}

One way of combining our gradient-based expectation value estimation algorithm with other approaches is to
apply the gradient-based expectation value estimation to some observables and to
measure others by sampling.
The purpose of this section is to analyze this trade-off and to show that there
are regimes where we achieve an overall scaling with \(\varepsilon\) that is
between \(\varepsilon^{-1}\) and \(\varepsilon^{-2}\).
In particular, we will consider the situation where $K$ (the number of groups of
mutually commuting operators, see \app{prior_estimation_work},
\eq{standard_approach_cost}) and $M$ are considered to be a function of $\varepsilon$ and
then ask when and how we should trade off between using gradient-based,
Heisenberg limited, estimates versus the shot-noise limited grouping strategy.
Under most circumstances, one of these two strategies will dominate the other.
However, in situations where there are a very large number of potential groups, this reasoning can change.

\subsection{Exponentially shrinking group sizes}
Let us consider dividing the observables into two groups: one that we estimate
using our gradient-based approach and one that we estimate using naive sampling.
Let \(M_{G}\) denote the number of observables in the first group and \(K\) denote
the number of groups of mutually commuting observables in the second group.
Then we find that the cost is
\begin{equation}
  C =\bigot{\frac{\sqrt{M_{PE}(\varepsilon)}}{\varepsilon} +
    \frac{K(\varepsilon)}{\varepsilon^2}} = \bigot{\frac{\sqrt{M -
        \sum_{k=1}^{K(\varepsilon)} T_k}}{\varepsilon} +
    \frac{K(\varepsilon)}{\varepsilon^2}},
\end{equation}
where  \(T_k\) denotes the cardinality of the \(k\)th group of mutually
commuting observables.
As a particular case to consider the scaling, let us take $\alpha>0, \Gamma>0$
be values such that
\begin{equation}
  T_k \le \Gamma Me^{-\alpha k}/(1-e^{-\alpha}).
\end{equation}
This case corresponds to the situation where most of the terms in the
Hamiltonian commute, but the grouping procedure becomes less efficient as we
consider more groups until the number of terms per group is at most $1$.

Thus the cost under this assumption scales as
$$
  C = \bigot{\frac{\sqrt{M e^{-\alpha K}}}{\varepsilon} + \frac{K}{\varepsilon^2}}.
$$
In this assignment, we see that our algorithm that is based on gradient
estimation is favorable to statistical sampling when $\sqrt{M} \varepsilon \le 1$.
However, if this is not true, then it is possible to find a minima for this
function.
Specifically, differentiating the cost with respect to $K$ and setting the
result to zero yields
$$
  K = \frac{\ln\left(\frac{\alpha^2{M} \varepsilon^2}{4}\right)}{\alpha}$$
Substituting the result gives $C = \bigot{\log(M)/\varepsilon^2 \alpha}$.
This suggests that for the case of exponentially shrinking group sizes, the
asymptotic scaling is identical to that of the sampling method alone; whereas if
$\varepsilon \sqrt{M} \ll 1$, no positive optimal $K$ exists and the extreme value
theorem suggests that the scaling abruptly shifts to the $\bigot{\sqrt{M}/
    \varepsilon}$ scaling predicted by our gradient method.

\subsection{Polynomially shrinking group sizes}

These trade-offs become more visible in cases where the size of each group of
commuting terms shrinks polynomially with $k$.
Specifically, let us assume that for some $\alpha > 1$, $\Gamma>0$ and any $K\ge
  1$,

$$
  T_k \ge \Gamma\frac{M}{\alpha-1} k^{-\alpha}
$$

From this we have from the fact that $1/k^{\alpha}$ is a convex function of $k$
that
\begin{equation}
  \sum_{k=1}^K T_k \ge \frac{\Gamma M}{\alpha-1}\int_{1}^K k^{-\alpha} \mathrm{d}k = M(1-K^{1-\alpha}),
\end{equation}

which approaches $M$ as $K\rightarrow \infty$ for $\alpha>1$.
This implies that the total cost obeys
\begin{equation}
  C = \bigot{\frac{\sqrt{M-\sum_{k=1}^K T_k }}{\varepsilon} + \frac{K}{\varepsilon}} = \bigot{\frac{\sqrt{MK^{1-\alpha}}}{\varepsilon} + \frac{K}{\varepsilon}}
\end{equation}
The first term above shrinks monotonically with $K$; whereas the second grows
monotonically.
This suggests that if a local optima exists, then it is a minima for $C$.
This local optima can be found by differentiating to find the best $K$, which
leads to
\begin{equation}
  K= e^{- \frac{\ln\left(\frac{4}{Me^2 (\alpha -1)^2}\right)}{1+\alpha}}.
\end{equation}
Substituting this value into our expression for $C$ yields
\begin{equation}
  C = \bigot{\frac{M^{\frac{1}{1+\alpha}}}{\varepsilon^{2(\frac{\alpha}{1+\alpha})}}}
\end{equation}
This shows that, depending on the falloff rate of the cummulative sum of the
sizes of the groups of commuting terms, intermediate scaling between the
Heisenberg-limited scaling of the gradient-based algorithm and the shot noise
limited scaling of the grouping can be observed.
Specifically, shot noise scaling is optimal as $\alpha \rightarrow \infty$ and
Heisenberg limited scaling occurs as $\alpha \rightarrow 1$.

\section{Estimating expectation values of observables with arbitrary norms}
\label{app:arbitrary_norms}

Theorem 1 in the main text applies only to collections of observables with
spectral norms at most $1$. Consider now an arbitrary collection of observables
$\{O_j\}_j$, each with a possibly different upper bound $B_j$ on its spectral
norm, $\|O_j\| \leq B_j$. The proof of Theorem~1 can be straightforwardly
extended to produce an algorithm that estimates the expectation values of these
observables using $\widetilde{\mathcal{O}}(B_{\max} \sqrt{M}/\varepsilon)$
queries to the state preparation oracle $U_\psi$ and its inverse, where
$B_{\max} \coloneqq \max_j B_j$. This is suboptimal whenever the $B_j$ are not
all equal to $B_{\max}$. In this appendix, we provide an algorithm that scales
with the 2-norm $\sqrt{\sum_j B_j^2}$ of the $B_j$'s rather than
$B_{\max}\sqrt{M}$, thereby proving the following generalisation of
Theorem~1.

\begin{theorem} \label{thm:ub_general} Let $\{O_j\}_{j=1}^M$ be an arbitrary collection
  of $M$ Hermitian operators on $N$ qubits, and for each $j$, let $B_j$ be a
  known upper bound on the spectral norm of $O_j$: $\|O_j\| \leq B_j$.
  There exists a quantum algorithm that, for any $N$-qubit quantum state
  $\ket{\psi}$ prepared by a unitary oracle $U_\psi$, outputs estimates
  $\widetilde{o_j}$ such that $|\widetilde{o_j} -
    \bra{\psi}O_j\ket{\psi}| \leq \varepsilon$ for all $j$ with probability at
  least $1-\delta$, using \[ Q = \mathcal{O}\left(\frac{\bar{B}}{\eps}
    {\log^{3/2}\left(\frac{\bar{B}}{\eps}\right)} \log\log
    \left(\frac{\bar{B}}{\eps}\right)\log\left(\frac{M}{\delta}\right) \right)
    = \widetilde{\mathcal{O}}\left(\frac{\bar{B}}{\varepsilon}\right) \] queries
  to $U_\psi$ and $U_\psi^\dagger$, where $\bar{B} \coloneqq \sqrt{\sum_{j\in M}
      B_j^2}$.
  The algorithm also uses $\bigo{Q\log(B_i/\varepsilon)}$ gates of the form controlled-$e^{-it O_j}$ for each $j \in \{1,\dots, M\}$, for various values of $t$ with $|t| \in \mathcal{O}\left(\bar{B}^{-1}\sqrt{\log(\bar{B}/\varepsilon)}\right)$, as well as
  $\mathcal{O}\left(\log(M/\delta)\sum\limits_{j\in[M]}\log(B_j/\eps)\log\log(B_i/\eps)
    + Q\log(Q/\eps)\log\log(Q/\eps)\right)$ elementary gates, and
  $\mathcal{O}\left(N + \sum_{j \in [M]}\log(B_j/\varepsilon) +
    \log\log(Q/\eps)\right)$ qubits.
\end{theorem}

Like the algorithm of Theorem~1 in the main text, a central idea is to encode the
expectation values in the gradient of the
function \begin{equation} \label{specialf} f(\bm{x}) \coloneqq
  -\frac{1}{2}\mathrm{Im} \Big[\bra{\psi} \prod_{j=1}^M e^{-2ix_j
        O_j}\ket{\psi}\Big] + \frac{1}{2}, \end{equation} a probability oracle for
which can be implemented as described in the main text.
However, the gradient estimation algorithm of \citen{Gilyen2017-gk} requires a
uniform upper bound on the gradient components, in order to guarantee the same
additive error $\varepsilon$ for each component.
This would lead to the suboptimal $\widetilde{\mathcal{O}}(B_{\max} \sqrt{M}/\varepsilon)$
scaling.
Therefore, to prove SI Theorem~\ref{thm:ub_general}, we must first generalize the
gradient estimation algorithm of \citen{Gilyen2017-gk} to allow for non-uniform
bounds on the gradient components.
We then analyze a different condition on the higher-order derivatives (from that
of Theorem~4 in the main text) that is more directly relevant to our
function $f$.

\subsection{generalized gradient estimation algorithm}

We recall some notation from \citen{Gilyen2017-gk}. For any $n \in \mathbb{N}$, they define the one-dimensional grid
$G_n \coloneqq \left\{{j}/{2^n} - {1}/{2} + {1}/{2^{n+1}} : j \in \{0,\dots, 2^n - 1\} \right\}$,
and use $\ket{x}$ for $x \in G_n$ to implicitly denote the state storing the binary representation of the integer $j$ where $x = j/2^n - 1/2 + 1/2^{n+1}$. For a function $f: \mathbb{R}^M \to \mathbb{R}$, a phase oracle $O_f$ for $f$ is any unitary that acts as $\ket{\bm{x}} \to e^{if(\bm{x})}\ket{\bm{x}}$, where $\bm{x} = \ket{x_1}\dots \ket{x_M}$ for $\bm{x} \in G_n^M$.

Our main modification to Algorithm~2 of \citen{Gilyen2017-gk} (i.e., Jordan's
gradient estimation algorithm~\cite{Jordan2005-hs}) is to allocate a possibly
different number of qubits $n_i$ to each $\ket{x_i}$ register.
Then $\bm{x} = \ket{x_1}\dots\ket{x_M}$ encodes some $\bm{x} \in G_{n_1} \times
  \dots \times G_{n_M}$, and we let
\[ \mathcal{G} \coloneqq G_{n_1} \times \dots \times G_{n_M} \]
denote this hyper-rectangular grid.
We state the modified algorithm for completeness. For now, $n_i$ and the parameter $S$, which determines the number of (fractional) queries to the phase oracle,
are all free parameters, to be determined later.

\begin{algorithm}[H] \caption{gradient estimation algorithm with variable fixed-point precision for each gradient component\label{alg1}}
  \textbf{Input:} A function $h : \mathcal{G} \to \mathbb{R}$, accessed via a phase oracle $O_h$ that acts as $O_h \ket{\bm{x}} = e^{i h(\bm{x})} \ket{\bm{x}}$ for all $\bm{x} \in \mathcal{G}$.
  \begin{algorithmic}[1]
    \State Initalise $\ket{0}^{\otimes n_1} \otimes \dots \otimes \ket{0}^{\otimes n_M}$.
    \State Apply the Hadamard transform to all registers, i.e., apply $H^{\otimes n_1} \otimes \dots \otimes H^{\otimes n_M}$.
    \State Apply $O_h^{2\pi S}$.
    \State For each $i \in [M] \coloneqq \{1,\dots, M\}$, apply the inverse quantum Fourier transform $\mathrm{QFT}_{G_{n_i}}^\dagger$ to register $i$, where for any $n \in \mathbb{N}$,
    \[ \mathrm{QFT}^\dagger_{G_{n}}\ket{\bm{x}} = \frac{1}{\sqrt{2^{n}}}\sum_{k \in G_{n}} e^{-2\pi i 2^{n} xk}\ket{k} \]
    for all $x \in G_{n}$.
    \State Measure in the computational basis; interpret the outcome as a vector $\bm{k} \in \mathcal{G}$.
  \end{algorithmic}
\end{algorithm}

\begin{lemma} \label{lem20}
  Let $N_i \coloneqq 2^{n_i}$ for all $i \in [M]$. Let $\bm{g} \in \mathbb{R}^M$ and $a,b,c \in \mathbb{R}$. If $h: \mathcal{G} \to \mathbb{R}$ is such that
  \[ |h(\bm{x}) - \bm{g}\cdot \bm{x} - c| \leq \frac{1}{a\pi S} \]
  for all but a $1/b$ fraction of the points $\bm{x} \in \mathcal{G}$, then the output $\bm{k}$ of Algorithm~\ref{alg1} satisfies
  \[ \Pr\left[\left|g_i -\frac{N_i}{S}k_i \right| \leq \frac{4}{S}\right] \geq \frac{2}{3} \]
  for each $i \in [M]$, provided that $N_i > 2 N|g_i|$ and $1/a^2 + 1/b \leq 1/2304$.
\end{lemma}
\begin{proof}[Proof sketch] This lemma is an extension of Lemma~20 of \citen{Gilyen2017-gk} (modified to remove the assumption that $\|\bm{g}\|_{\infty} \leq 1/3$), and can be proven by straightforwardly adapting the proof therein, so we highlight the main differences. The ``ideal'' state after the inverse Fourier transforms in Step~4 of Algorithm~\ref{alg1} is (up to a global phase)
  \[ \bigotimes_{i \in [M]} \left[\frac{1}{N_i}\sum_{x_i,k_i \in G_{n_i}} e^{2\pi i N_i x_i \left( \frac{S}{N_i}g_i - k_i\right)}\ket{k_i} \right], \]
  so the analysis of phase estimation in~\cite{NCbook} shows that for any $\kappa > 1$, the output $\bm{k}$ satisfies
  \begin{equation} \label{qpe_analysis} \Pr\left[\left|\frac{S}{N_i}g_i - k_i \right| > \frac{\kappa}{S} \right] \leq \frac{1}{2(\kappa - 2)}. \end{equation}
  The rest of the proof of Lemma~20 in \citen{Gilyen2017-gk} follows through
  with simple modifications. (One minor technical point is that in
  \citen{Gilyen2017-gk}, the bound in SI Eq.~\eqref{qpe_analysis} was mistakenly
  stated with $1/[2(\kappa - 1)]$ on the right-hand side, but this does not take
  into account fixed-point approximation error. For this reason, in
  \citen{Gilyen2017-gk} the parameters $a$ and $b$ are set to $42$ and $1000$,
  respectively, but with the corrected bound we find that $a$ and $b$ need to
  satisfy $1/a^2 + a/b \leq 1/(24)^2/4 = 1/2304$.) Finally, note from the standard
  analysis of phase estimation that the measurement outcome can be unambiguously
  interpreted to determine the (approximate) value of $g_i$ provided that the
  range of $Sg_i/N_i$ is less than $1$. For this, it suffices to have $S|g_i|/ N_i
    < 1/2$, i.e., $N_i > 2S|g_i|$.
\end{proof}

We can now derive the analogue of Theorem~21 in \citen{Gilyen2017-gk}.

\begin{lemma} \label{lem21}
  Let $\bm{g},\bm{y} \in \mathbb{R}^M$,  $c \in \mathbb{R}$, and $r, \delta, \varepsilon \in \mathbb{R}_+$. Suppose that $\widetilde{f}: r \cdot \mathcal{G} \to \mathbb{R}$ is such that
  \begin{equation} \label{fcondition} |\widetilde{f}(r\bm{x} + \bm{y}) - r\bm{x} \cdot \bm{g} - c| \leq \frac{\varepsilon r}{8 a\pi} \end{equation}
  for all but a fraction $1/b$ of the points $\bm{x} \in \mathcal{G}$, with $1/a^2 + 1/b \leq 1/2304$, and that for all $i \in [M]$, we have $|g_i| \leq z_i$ for some $z_i \in \mathbb{R}_+$. Assume access to a phase oracle $O_{\widetilde{f}} : \ket{\bm{x}} \mapsto e^{i \widetilde{f}(r\bm{x} + \bm{y})}\ket{\bm{x}}$. Then, we can compute a $\widetilde{\bm{g}}$ such that
  \begin{equation} \label{gtilde} \Pr\left[\|\bm{g} - \widetilde{\bm{g}}\|_{\infty} \leq \varepsilon\right] \geq 1-\delta\end{equation}
  using $\mathcal{O}(\log(M/\delta))$ queries to $O_{\widetilde{f}}^{2\pi S}$ with $S = 4/(\varepsilon r)$, $\mathcal{O}\left(\log(M/\delta)\sum_{i \in [M]} \log(z_i/\varepsilon)\log\log(z_i/\varepsilon) \right)$ gates, and $\sum_{i \in [M]} \lceil \log(12 z_i/\varepsilon)\rceil$ qubits.
\end{lemma}
\begin{proof}[Proof sketch]
  We set $n_i = \lceil \log(12 z_i/\varepsilon)\rceil$ for each $i$, $S = 4/(\varepsilon r)$, and $h(\bm{x}) = \widetilde{f}(r\bm{x}+ \bm{y})$ in Algorithm~\ref{alg1}. Then, \[ |h(\bm{x}) - \bm{x} \cdot r \bm{g}| = | \widetilde{f}(r\bm{x} + \bm{y}) - r\bm{x}\cdot \bm{g}| < \frac{\varepsilon r}{8a\pi} \leq \frac{1}{a\pi S}, \]
  and the other conditions of SI Lemma~\ref{lem20} are satisfied as well, so the output $\bm{k}$ of Algorithm~\ref{alg1} satisfies
  \[ \left|g_i - \frac{1}{r}\frac{N_i}{S} k_i\right| \leq \frac{4}{rS} =
    \varepsilon\] with probability at least $2/3$, for each $i$.
  Therefore, by repeating Algorithm~\ref{alg1} $\mathcal{O}(\log(M/\delta))$ times and
  taking the median~\footnote{In \citen{Gilyen2017-gk}, the median is computed using a quantum circuit. For our purposes, it suffices to repeat the circuit $\mathcal{O}(\log(M/\delta))$ times, measure at the end of Algorithm~\ref{alg1}, and take the median of the measurement outcomes classically.} of the outputs gives a $\widetilde{\bm{g}}$ satisfying
  SI Eq.~\eqref{gtilde}.
  The gate complexity is dominated by that of the inverse Fourier transforms.
\end{proof}

\subsection{Condition on the higher-order derivatives} With SI Lemma~\ref{lem21} in hand, it remains to find a function $\widetilde{f}$ satisfying SI Eq.~\eqref{fcondition}  where $\bm{g}$ the gradient of the function $f$ that we are interested in. To satisfy the analogous condition in their setting, \citen{Gilyen2017-gk} takes $\widetilde{f}$ to be a degree-$2m$ central difference formula. They then make the assumption that for all $k \in \mathbb{N}$, the $k$th derivatives of $f$ satisfy  $|\partial_\alpha f(\bm{0})| \leq c^k k^{k/2}$, and show that this implies $|f_{(2m)}(\bm{y}) - \bm{y}\cdot \nabla f(\bm{0})| \leq \sum_{k=2m+1}^\infty (8rcm \sqrt{M})^k$ for all but a fraction $1/b$ of $\bm{y} \in r \cdot G_n^d$ (for a $b$ such that $1/b < 1/2304$). They then choose $r$ such that $r^{-1} =\mathcal{O}(cm\sqrt{M}(acm\sqrt{M}/\varepsilon)^{1/(2m)})$ upper-bound the right-hand side by $\varepsilon r/(8a\pi)$.

Of course, the $|\partial_\alpha f(\bm{0})| \leq c^{k} k^{k/2}$ assumption is not suitable for our purposes, where the gradient components are allowed to have different magnitudes. Motivated by the function $f$ constructed for the expectation value estimation algorithm (SI Eq.~\eqref{specialf}), we instead make the following assumption: there exists some $\bm{z} \in \mathbb{R}^M$ such that for all $k \in \mathbb{N}$,
\begin{equation} \label{derivatives}
  |\partial_\alpha f(\bm{0})| \leq \bm{z}^\alpha
\end{equation}
for all $\alpha \in [M]^k$ (where for $\alpha = (\alpha_1, \dots, \alpha_k) \in [M]^k$, we have $\partial_{\bm{\alpha}} f \equiv \partial_{\alpha_1}\dots \partial_{\alpha_k} f$ and $\bm{z}^{{\bm{\alpha}}} \equiv z_{\alpha_1}\dots z_{\alpha_k}$). We now sketch the proof that SI Eq.~\eqref{derivatives} implies that $|f_{(2m)}(\bm{y}) - \bm{y} \cdot \nabla f(\bm{0})| \leq \sum_{k = 2m+1}^\infty (8\|\bm{z}\|rm/\sqrt{k})^k$, and the rest of the analysis of \citen{Gilyen2017-gk} essentially follows through with $c\sqrt{M}$ replaced by the 2-norm $\|\bm{z}\|$ of $\bm{z}$, and $m$ replaced by $\sqrt{m}$.

\begin{lemma} \label{lem24} Let $r \in \mathbb{R}_+$. Suppose that $f: \mathbb{R}^M \to \mathbb{R}$ is analytic and that there exists a $\bm{z} \in \mathbb{R}^M$ such that for all $k \in \mathbb{N}$ and ${\bm{\alpha}} \in [M]^k$, we have $|\partial_{{\bm{\alpha}}} f(\bm{0})| \leq \bm{z}^{{\bm{\alpha}}}$. Then \[ |f_{(2m)}(\bm{y}) - \bm{y} \cdot \nabla f(\bm{0})| \leq \sum_{k=2m+1}^\infty \left(\frac{8\|\bm{z}\|rm}{\sqrt{k}}\right)^k \] for all but a $1/3840$ fraction of points $\bm{y} \in r\cdot \mathcal{G}$.
\end{lemma}
\begin{proof}[Proof outline.]
  The key step is to modify Proposition~13 of \citen{Gilyen2017-gk} to apply to the derivative condition SI Eq.~\eqref{derivatives}. In particular, consider drawing $\bm{y} \in \mathcal{G}$ uniformly at random. Then, the components $y_1, \dots, y_M$ are i.i.d. symmetric random variables bounded in $[-1/2,1/2]$, and we have
  \begin{align*}
    \E\Bigg[\Bigg(\sum_{{\bm{\alpha}} \in [M]^k} \bm{y}^{{\bm{\alpha}}} \partial_{{\bm{\alpha}}} f(\bm{0}) \Bigg)^2 \Bigg] & = \sum_{{\bm{\alpha}},{\bm{\beta}} \in [M]^k} \E[\bm{y}^{\bm{\alpha}} \bm{y}^{\bm{\beta}}] \partial_{\bm{\alpha}} f(\bm{0})\partial_{\bm{\beta}} f(\bm{x}) \\
                                                                                                                           & \leq \sum_{{\bm{\alpha}},{\bm{\beta}} \in [M]^k} \E[\bm{y}^{\bm{\alpha}} \bm{y}^{\bm{\beta}}] \bm{z}^{\bm{\alpha}} \bm{z}^{\bm{\beta}}                     \\
                                                                                                                           & = \E\Bigg[\Bigg(\sum_{i \in [M]} y_i z_i\Bigg)^{2k} \Bigg]                                                                                                 \\
                                                                                                                           & = \int_0^\infty dt\, \Pr \Bigg[\Bigg|\sum_{i \in [M]}y_i z_i \Bigg| \geq t^{1/(2k)} \Bigg]                                                                 \\
                                                                                                                           & \leq \int_0^\infty dt\, 2\exp\left(-\frac{2t^{1/k}}{\sum_{i \in [M]} z_i^2} \right)                                                                        \\
                                                                                                                           & = 2 \left(\frac{\|\bm{z}\|^2}{2}\right)^k k!                                                                                                               \\
                                                                                                                           & < 2\left( \frac{\|\bm{z}\|^2 k}{2}\right)^k,
  \end{align*}
  where the second inequality follows from Hoeffding's inequality, using the fact that $y_i z_i \in [-z_i/2, z_i/2]$. Hence, by Markov's inequality, we have
  \begin{equation} \label{newbound} \Bigg|\sum_{{\bm{\alpha}} \in [M]^k} \bm{y}^{{\bm{\alpha}}} \partial_{\bm{\alpha}} f(\bm{0}) \Bigg| \leq \sqrt{2}(4\|\bm{z}\|\sqrt{k/2})^k \end{equation}
  for all but at most a $1/4^{2k}$ fraction of points $\bm{y} \in \mathcal{G}$. Now, comparing SI Eq.~\eqref{newbound} to Equation~52 of \citen{Gilyen2017-gk}, we see that the rest of the proof of Theorem~24 in \citen{Gilyen2017-gk} goes through with $c \sqrt{M}$ replaced by $\|\bm{z}\|/\sqrt{k}$, leading to the claimed result.
\end{proof}
The remaining step is to choose $r$ so that $f_{(2m)}$ satisfies the conditions required of $\widetilde{f}$ in SI Lemma~\ref{lem21}. This gives our analogue of \citen{Gilyen2017-gk} Theorem~25 (see also Theorem~4 in the main text).
\begin{theorem} \label{thm:gradient_general}
  Let $\bm{x} \in \mathbb{R}^M$. Suppose that $f: \mathbb{R}^M \to \mathbb{R}$ is analytic and that there exists a $\bm{z} \in \mathbb{R}^M$ such that for all $k \in \mathbb{N}$, ${\bm{\alpha}} \in [M]^k$, we have $|\partial_{\bm{\alpha}} f(\bm{x})| \leq \bm{z}^{\bm{\alpha}}$. Then, for any $0 < \varepsilon < \|\bm{z}\|$, we can compute a $\widetilde{\bm{g}}$ such that $\Pr[|\nabla f(\bm{0}) - \widetilde{\bm{g}}\|_{\infty} \leq \eps] \geq 1- \delta$ using
  \[ \mathcal{O}\left(\left[\frac{\|\bm{z}\|}{\eps} \sqrt{\log\frac{\|\bm{z}\|}{\eps}} \log\log \frac{\|\bm{z}\|}{\eps}\right]\log\frac{d}{\delta} \right) \]
  (fractional) queries to a phase oracle $O_f$ for $f$, $\mathcal{O}\left(\log(M/\delta)\sum\limits_{i\in[M]}\log(z_i/\eps)\log\log(z_i/\eps)\right)$ gates, and $\mathcal{O}\left(\sum_{i \in [M]}\log(z_i/\varepsilon)\right)$ qubits.
\end{theorem}
\begin{proof}[Proof sketch]
  This can be proven using the same arguments as in the proof of Theorem~25 in \citen{Gilyen2017-gk}. The main difference is that we have a different bound on $|f_{(2m)}(\bm{y}) - \bm{y}\cdot \nabla f(\bm{0})|$ from SI Lemma~\ref{lem24}. Consequently, instead of setting $r$ as in \citen{Gilyen2017-gk}, we choose $r$ so that
  \[ r^{-1} = 9\|\bm{z}\| \sqrt{m/2} (64 \cdot 8a\pi \|\bm{z}\| \sqrt{m/2}/\varepsilon)^{1/(2m)}. \]
  Proceeding through the rest of the proof in \citen{Gilyen2017-gk} with the appropriate modifications leads to the stated query complexity. The gate and qubit complexities follow directly from Lemma~\ref{lem21}, observing that the assumption on the derivatives implies that $\partial_i f(\bm{x}) \leq z_i$ for every $i \in [M]$.
\end{proof}

We can now prove our general theorem, SI Theorem~\ref{thm:ub_general}, for estimating the expectation values of arbitrary observables, by showing that the particular function $f$ in SI Eq.~\eqref{specialf} whose gradient we are interested in satisfies the condition of SI Theorem~\ref{thm:gradient_general}.
\begin{proof}[Proof of SI Theorem~\ref{thm:ub_general} (sketch)]
  Let $f$ be the function defined in SI Eq.~\eqref{specialf}. As shown in the main text, the components of $\nabla f(\bm{0})$ are exactly the expectation values $\bra{\psi}O_j\ket{\psi}$, and a probability oracle for $f$ can be constructed using one query to each of $U_\psi$ and $U_\psi^\dagger$. Note that for any $k \in \mathbb{N}$ and ${\bm{\alpha}} \in [M]^k$, we have
  \[ |\partial_{\bm{\alpha}} f(\bm{0})| = \Big|-\frac{1}{2} \mathrm{Im}\Big[\bra{\psi} \prod_{i \in [k]} (-2O_{\alpha_i})\ket{\psi} \Big] \Big| \leq \prod_{i \in [k]} (2 \|O_{\alpha_i}\|) \leq \prod_{i \in [k]} (2B_{\alpha_i}) \]
  (where the order of the $O_{\alpha_i}$'s in the second expression depends on the values of the $\alpha_i$'s), so we can take $\bm{z} = 2(B_1,\dots, B_M)$ in SI Theorem~\ref{thm:gradient_general} to obtain the gate and qubit complexities, as well the number of queries $Q$ to a phase oracle for $f$. By Theorem~14 of \citen{Gilyen2017-gk}, an $\varepsilon'$-approximate phase oracle for $f$ can be implemented using $\mathcal{O}(\log(1/\varepsilon'))$ queries to a probability oracle for $f$, $\mathcal{O}(\log(1/\varepsilon')\log\log(1/\varepsilon'))$ gates, and $\mathcal{O}(\log\log(1/\varepsilon'))$ ancillas. Setting $\varepsilon' = x(\varepsilon/Q)$ gives the result.
\end{proof}

\end{document}